\documentclass[aps,showpacs, preprint]{revtex4}

\usepackage{amsfonts,amssymb,bbold}
\usepackage{amsmath,amsthm,mathtools}
\usepackage[pdftex,letterpaper=true,pagebackref=false]{hyperref}
\usepackage{verbatim} 

\newtheorem{theorem}{Theorem}[section]
\newtheorem{lemma}[theorem]{Lemma}

\theoremstyle{definition}

\theoremstyle{remark}
\newtheorem{remark}[theorem]{Remark}

\DeclareMathOperator{\tr}{Tr}

\def\>{\rangle}
\def\rrangle{\rangle\!\rangle}
\def\<{\langle}
\def\llangle{\langle\!\langle}

\def\eps{\epsilon}

\def\bbN{\mathbb{N}}

\def\bbR{\mathbb{R}}

\def\bbZ{\mathbb{Z}}
\def\bbone{\mathbb{1}}

\def\cA{\mathcal{A}}

\def\cH{\mathcal{H}}
\def\cI{\mathcal{I}}

\def\cK{\mathcal{K}}

\def\cM{\mathcal{M}}
\def\cN{\mathcal{N}}

\def\cR{\mathcal{R}}

\def\cU{\mathcal{U}}
\def\cV{\mathcal{V}}
\def\cX{\mathcal{X}}
\def\cY{\mathcal{Y}}

\def\ru{{\rm u}}
\def\rv{{\rm v}}
\def\rx{{\rm x}}
\def\ry{{\rm y}}

\def\fB{\mathfrak{B}}
\def\fD{\mathfrak{D}}

\def\fK{\mathfrak{K}}

\def\tH{ {\bf H} }
\def\tK{ {\bf K} }

\def\tU{ {\bf U} }

\def\eps{\epsilon}

\def\w{\omega}

\def\half{\frac{1}{2}}

\def\tens#1{{\bf {#1}}}

\begin{document}

\title{Truncated quantum channel representations for coupled harmonic oscillators}

\author{Yingkai Ouyang }
\affiliation{\footnotesize Department of Combinatorics and Optimization, Institute of Quantum Computing, University of Waterloo, \\ 200 University Avenue West,
\footnotesize Waterloo, Ontario N2L 3G1, Canada.\\ \footnotesize \texttt{y3ouyang@math.uwaterloo.ca}}
\author{Wee Hao Ng}
\affiliation{\footnotesize Department of Physics,
Cornell University,
109 Clark Hall,
Ithaca, New York 14853-2501, USA \\ \footnotesize \texttt{wn68@.cornell.edu}}

\begin{abstract}
Coupled quantum harmonic oscillators, studied by many authors using many different techniques over the decades, are frequently used toy-models to study open quantum systems. In this manuscript, we explicitly study the simplest oscillator model -- a pair of initially decoupled quantum harmonic oscillators interacting with a spring-like coupling, where the bath oscillator is initially in a thermal-like state. In particular, we treat the completely positive and trace preserving map on the system as a quantum channel, and study the truncation of the channel by truncating its Kraus set and its output dimension. 
We thereby derive the 
 truncated transition amplitudes of the corresponding truncated channel. Finally, we give a computable approximation for these truncated transition amplitudes with explicit error bounds, and perform a case study of the oscillators in the off-resonant and weakly-coupled regime numerically.
 We demonstrate explicitly that the substantial leakage error can be mitigated via quantum error correction.
\end{abstract}
\pacs{03.67a, 03.65Yz}

\maketitle

%


%
%
%
%

\section{Introduction}

One of the canonical physical models in quantum physics is that of quantum oscillators coupled with harmonic baths.
The dynamics of such models and their variations has been extensively studied, using various techniques \cite{FKM65,RLL67,EKN68,BBW73,Dav73,MiH86,NRSS09,YUKG88,HPZ92,CYH08,MaG12}. These techniques include Markovian master equations \cite{Dav74}, quantum stochastic processes and quantum Langevin equations \cite{Dav69,Dav70,Dav71,FKM65,RMa81,FLO88}, Kossakowski-Lindblad equations \cite{Kos72,Lin76}, methods in density-functional theory \cite{parr1994density}, the standard techniques of perturbation theory, among many others \cite{Kos83,LCDFGZ87}.

Quantum channels \cite{nielsen-chuang} can be used to quantify the dynamics of a quantum system,
and can be described by truncated transition amplitudes.
In this paper, we approximate the truncated transition amplitudes of a given channel, where the truncation is performed with respect to the quantum channel's set of Kraus operators and its dimensions. 
A truncated quantum channel is a trace-decreasing quantum operation, and quantifies the partial dynamics acting on the system. Knowledge of the truncated quantum channel has utility -- lower bounds on the performance of quantum error correction codes with its knowledge \cite{LNCY97,BaK02,Fletcher08,Kosut08,BaG09,Tys10,BeO10,BeO11,Ouyang-PKL}.

In this paper, we work towards quantifying the approximate dynamics of a pair of initially decoupled quantum harmonic oscillators interacting with a spring-like coupling, where the bath oscillator is initially in a thermal-like state. We work with a truncated subset of the model's Kraus operators, and thereby approximate its truncated transition amplitudes. We note that the Kraus operators of oscillator-bath models have also been approximated by various authors 
\cite{MiH86,CLY97,Liu04}. Recently, Holevo also gave a formal exact expression for the Choi-Jamiolkowski operator for Gaussian channels \cite{Hol11}, which describes the dynamics of coupled oscillators. Our contributions in this paper, are the explicit upper bounds on the approximation error of the truncated transition amplitudes of two quantum harmonic oscillators coupled via a spring-like interaction, where the approximation is an explicit summation of a finite number of computable terms,
and dependent on the size of the input and output dimensions of the truncated channel. Our results can be used to explicitly study this toy model with rigorous error bounds. 
In particular, we numerically demonstrate and provide lower bounds for the leakage error, and show how this leakage error is mitigated via quantum error correction.

The organization of the paper is as follows. In Section \ref{sec:preliminaries}, we introduce the preliminary material needed for this paper. In particular, we review $L^2(\bbR)$ Hilbert space, quantum states, quantum channels, Hermite functions, and the linear canonical transformations for the quantum harmonic oscillator. In Section \ref{sec:model-description}, we give a treatment of the truncated dynamics of two quantum harmonic oscillators interacting with a spring-like coupling, and give explicit bounds on the error term induced by approximating the truncated transition amplitudes with a finite sum in Theorem \ref{thm:qho-error-bound}. In Section \ref{sec:hermite-function-bounds} we give bounds on Hermite functions that are needed for the proof of Theorem \ref{thm:qho-error-bound}. Finally we apply our results explicitly in Section \ref{sec:example} in the case where the oscillators are off-resonant and weakly coupled. 

\section{Preliminaries} \label{sec:preliminaries}
In this section, we review the theory of $L^2(\bbR)$ Hilbert spaces, quantum states and various representations of quantum channels, Hermite polynomials and functions, and coupled harmonics oscillators.
\subsection{The $L^2(\bbR)$ Hilbert spaces}
We refer the reader to \cite{reed+simon-I} and \cite{Bogo+L+T} for an introduction to separable Hilbert spaces.
 In this paper, all Hilbert spaces are complex and separable, have countable bases, and are typically isomorphic to the set of square integrable functions $L^2(\bbR)$ which necessarily have infinite dimensions. Let $\cH^*$ be the dual space of Hilbert space $\cH$. For Hilbert spaces $\mathcal H$ and $\cK$, let $L(\mathcal H, \cK)$ denote the set of linear operators mapping $\mathcal H$ to $\mathcal K$, and let $L(\cH) := L(\cH,\cH)$. Let $\fB(\cH, \cK)$ denote the set of bounded operators in $L(\cH, \cK)$. 

We use the Dirac's `ket' $|\psi_\cH\>$ to denote a function $\psi$ in the Hilbert space $\cH$. For Hilbert space $\cH$, 
let $\{|j_\cH\> \}_{j \in \bbN}$ denote its generic orthonormal basis. We also use $|\psi_\cH, \varphi_\cK\>$ and $|\psi_\cH\>| \varphi_\cK\>$ to denote $|\psi_\cH\> \otimes |\varphi_\cK\>$. We denote the inner product of Hilbert space $\cH$ using the Dirac notation as $\< \cdot | \cdot \>_\cH $ which is a sesquilinear form. 
For functions $\psi, \varphi \in \cH = L^2(\bbR)$, their inner product in the position basis is $\<\psi|\varphi\>_\cH := \int_\bbR \psi(x)^* \varphi(x) dx $. 
The physicist's position and momentum operators $ \hat x_\cH$ and $\hat p_\cH$ are sesquilinear forms that map the tuple $(|\psi_\cH\>, |\varphi_\cH\>)$ to 
$\int_\bbR \psi(x)^* x \varphi(x) dx
=\<\psi_\cH | \hat x_\cH| \varphi_\cH\>
$ and 
  $\int_\bbR \psi(x)^*  \frac{\hbar \partial}{i\partial x}\varphi(x) dx
  =
  \<\psi_\cH | \hat p_\cH| \varphi_\cH\>$
   respectively. We also denote the identity map on $\cH$ as $\bbone_\cH$.
\subsection{Quantum states and channels}
We refer the reader to \cite{nielsen-chuang} for an introduction to quantum states and channels. Define the set of quantum states on Hilbert space $\mathcal H$ to be $\mathfrak D(\mathcal H)$, the set of all positive semi-definite and trace one operators in $\fB (\cH)$. When $\rho \in \fD(\cH \otimes \cK)$, we denote the partial trace of $\rho$ on Hilbert space $\cH$ as
$\tr_\cH(\rho ):= \<j_\cH| \rho |j_\cH\>$.

A quantum channel $\Phi : \fB(\mathcal H) \to \fB(\mathcal K)$ is a completely positive and trace-preserving (CPT) linear map, and its non-unique Kraus representation is \cite{HeK69,HeK70,kraus}
\[\Phi(\rho) = \sum_{\tens{K} \in \fK} \tens{K} \rho \tens{K}^\dagger, \quad
\sum_{\tens{K} \in \fK} \tens{K}^\dagger \tens{K} = \mathbb 1_{\cH} \]
where $\fK \subset \fB(\cK, \cH)$ is called the Kraus set of $\Phi$. We denote the basis-dependent matrix elements of the Kraus operators by $\tK_{j,j'}$ so that for all $\tK  \in \fK$,
\begin{align*}
\tK = \sum_{j,j'} \tK_{j,j'} |j'_\cK\>\<j_\cH|.
\end{align*}
 We define the {\bf transition amplitudes} of $\Phi$ with respect to the Kraus set $\fK$ to be
\begin{align}
T^{(a,b) \to (a',b')}_{\fK} := \sum_{\tK \in \fK}   \tK_{b,b'}\tK_{a,a'}^*,  \label{eq:transition-amplitude-def}
\end{align}
a sum of the product of two Kraus operators over the entire Kraus set.
Now let $\rho \in \fD(\cH)$ and $\Phi(\rho) \in \fD(\cK)$ have the decompositions
\begin{align*}
\rho = \sum_{a,b} \rho_{a,b} |b_\cH\> \<a_\cH| , \quad \Phi(\rho) = \sum_{a,b} \rho'_{a',b'} |b'_\cK\> \<a'_\cK|
 \end{align*}
so that in the Kraus representation,
\begin{align}
&\< b'_\cK  | \Phi(\rho) | a'_\cK \>
=  \< b'_\cK  |
            \sum_{\tK \in \fK } \tK \rho \tK^\dagger
        | a'_\cK \> \notag\\
    =&  \< b'_\cK  |
            \sum_{\tK \in \fK } \sum_{ j,j' } \tK_{j,j'} |j'_\cK\>\<j_\cH|
            \sum_{a,b}  \rho_{a,b}  |b_\cH\>\<a_\cH|
            \sum_{ k, k' } \tK_{k, k'}^* | k_\cH\>\< k'_\cK| a'_\cK \> \notag\\
    =&  \sum_{\tK \in \fK } \sum_{ j,j' }  \sum_{a,b} \sum_{ k, k' }
            \tK_{j,j'} \tK_{k, k'}^* \rho_{a,b}
            \< b' | j'\>_\cK \<j |b\>_\cH \<a | k\>_\cH \< k'| a'\>_\cK \notag\\
    =&   \sum_{a,b\in \bbN}  \left( \sum_{\tK \in \fK }
            \tK_{b,b'} \tK_{a, a'}^*\right) \rho_{a,b} =
            \sum_{a,b\in \bbN}
            T^{(a,b) \to (a',b')}_{\Phi, \fK}	\rho_{a,b}.
    \label{eq:kraus-finalstate}
\end{align}
Hence $T^{(a,b) \to (a',b')}_{\fK}$ quantifies the transition amplitudes of $\<b_\cH|\rho|a_\cH\>$ to $\<b'_\cK|\Phi(\rho)|a'_\cK\>$.
For the purpose of quantum information processing, it may not be necessary to work with the full Kraus set $\fK$. In this paper, we instead restrict our attention to the {\it truncated Kraus set} $\Omega$, which is some appropriately chosen subset of the full Kraus set. This truncation procedure approximates the channel well if the truncated Kraus set comprises of the `typical' Kraus effects. For the purpose of quantum error correction, partial knowledge of the channel is already of great utility, and recovery channels can be constructed based on this partial information to give lower bounds on the entanglement fidelity of specifically chosen quantum codes. Hence in this paper, the {\it truncated transition amplitudes} $T^{(a,b) \to (a',b')}_{\Omega}$ play a central role in quantifying the truncated dynamics of the channel $\Phi$.
\subsection{Hermite polynomials and functions}
For $n \in \bbN$, define the Hermite polynomials $H_n(x)$ and the Hermite functions $\psi_n(x)$ to be
\begin{align}
H_n(x) &:= (-1)^n e^{x^2} \frac{d^n }{dx^n} e^{-x^2} , \quad
\psi_n(x) := \frac{e^{-\frac{1}{2}x^2} H_n(x) }{\sqrt{2^n  n! \sqrt{\pi} }}.
\end{align}
For example, $H_0(x) = 1$ and $H_1(x) = 2x$.
The reader can refer to \cite{abramowitz+stegun} for the properties of the Hermite polynomials and functions. For $c> 0 $, also define the rescaled Hermite function to be
\begin{align}
\psi_{n,c}(x) = \<\tilde x_\cH | \psi_{n,c, \cH}\> := \sqrt{c} \psi_n(c x) .
\end{align}
As stated by Watson \cite{Wat33}, Mehler's formula applies in the case when $|z| < 1$ and $z$ is real, that is
\begin{align}
\sum_{n = 0}^\infty z^n \psi_n(x) \psi_n(y) = \frac{1}{\sqrt{\pi(1-z^2)} }
            \exp \left[
                \frac{4xy z - (x^2 + y^2)(1+z^2) }{2(1-z^2)}
            \right] . \label{eq:mehler2}
\end{align}
Mehler's formula also holds for all complex numbers $|z| < 1$, with the series converging uniformly and absolutely [Theorem 23.1 \cite{wong}].
\subsection{A pair of harmonic oscillators and their linear canonical transformations}
\subsubsection{The classical model}
We refer the reader to \cite{shankar} for an introduction to the quantum harmonic oscillator.
Define the classical Hamiltonian of a classical harmonic oscillator with mass $m$, resonant frequency $\omega$, position coordinate $x$ and momentum coordinate $p$ to be
\begin{align}
H_{ m,\omega ; x,p} := \frac{p^2}{2m} + \half m \omega^2 x^2.
\end{align}
The model we study has the classical Hamiltonian
\[
H = H_{ m_\rx, \omega_{\rx, \bf o} ; x, p_\rx  } +
    H_{ m_\ry, \omega_{\ry, \bf o} ; y, p_\ry  } +
    H_{{\rm int}, {\bf o}}.
\]
where $H_{{\rm int}, {\bf o}} := \frac{1}{2}k (x-y)^2 $ is the classical Hamiltonian representing the spring-like interaction between the oscillators where $k \ge 0$.
The spring-like interaction $H_{{\rm int}, {\bf o}}$ introduces quadratic terms 
$\frac{k x^2}{2}$ and 
$\frac{k y^2}{2}$ into $H$, effectively renormalizing the oscillator frequencies from 
$\omega_{\rx, {\bf o}}$ and 
$\omega_{\ry, {\bf o}}$ to 
$\omega_\rx := \sqrt{\omega_{\rx, {\bf o}}^2 
+ \frac{k}{m_\rx}}$ and 
$\omega_\ry := \sqrt{\omega_{\ry, {\bf o}}^2 
+ \frac{k}{m_\ry}}$ respectively. 
Hence when 
$ H_{\rm int} := - k  x y$,
\begin{align}
H =
    H_{ m_\rx, \omega_{\rx} ; x, p_\rx  } +
    H_{ m_\ry, \omega_{\ry} ; y, p_\ry  } +
    H_{\rm int} \label{eq:renorm} .
\end{align}
In an experimental setup, it may be impossible to turn off the interaction between the two oscillators. Then the physically measured oscillator frequencies correspond to the renormalized frequencies. Therefore, we work with the renormalized representation of the model Hamiltonian given by (\ref{eq:renorm}).
\subsubsection{The quantized model} \label{subsubsec:quantized-model}
Define the Hamiltonian of a quantum harmonic oscillator with associated Hilbert space $\cH$, mass $M>0$, resonant frequency $\omega>0$, position operator $\hat x_\cH$ and momentum operator $\hat p_\cH$ to be
\begin{align}
\tH_{ M,\omega ; \hat x_\cH, \hat p_\cH } := \frac{\hat p_\cH^2}{2M} + \half M \omega^2 \hat x_\cH^2.
\end{align}
The set of rescaled Hermite functions $\{| \psi_{n, \sqrt{\frac{M\omega}{\hbar}}, \cH } \> \}_{n \in \bbN}$ is the set of energy eigenfunctions of the Hamiltonian $\tH_{(M,\omega ; \hat x_\cH, \hat p_\cH)}$.
Let the Hilbert space of the first and second oscillators be $\mathcal X$ and $\mathcal Y$ respectively, both isomorphic to $L^2(\bbR)$. Define $\hat x := \hat x_\cX \otimes \bbone_\cY$, $\hat y := \bbone_\cX \otimes \hat x_\cY$, $\hat p_x := \hat p_\cX \otimes \bbone_\cY$ and $\hat p_y := \bbone_\cX \otimes \hat p_\cY$.
Then the quantized model Hamiltonian (\ref{eq:renorm}) is
\begin{align}
\tH &=
    \tH_{ m_\rx, \omega_{\rx} ; \hat x_\cX, \hat p_\cX  } \otimes \bbone_\cY +
    \bbone_\cX \otimes \tH_{ m_\ry, \omega_{\ry} ; \hat x_\cY, \hat p_\cY  } +
    \tH_{\rm int} \notag\\
    &=
    \tH_{ m_\rx, \omega_{\rx} ; \hat x, \hat p_x  }  +
    \tH_{ m_\ry, \omega_{\ry} ; \hat y, \hat p_y  } +
    \tH_{\rm int}
      \label{eq:quantum-renorm}
\end{align}
where $\tH_{\rm int} := - k \hat x \hat y$ is the quantized interaction.

The coupled quantum harmonic oscillators can be decoupled by a linear canonical transformation of the oscillator positions and momenta \cite{JMM11}.
 Define the rotation matrix, the rotation angle, and the rescaled mass by
\begin{align}
{\bf R} := \footnotesize \begin{pmatrix} \cos \theta & \sin \theta \\ -\sin \theta & \cos \theta  \end{pmatrix} , \quad
 \quad \theta := \frac{1}{2}\tan^{-1} \left( \frac{2k/m}{\omega_\ry^2 - \omega_\rx^2} \right) ,
 \quad m := \sqrt{m_\rx m_\ry}
 \end{align}
 respectively. The use of straightforward trigonometry then gives 
\begin{align}
\cos \theta =
 \frac{1}{\sqrt{2}} \left( 1 +
 \frac{1}
         {\sqrt{      1+     \frac{4k^2/m^2}
                          {(\omega_\ry^2 - \omega_\rx^2)^2}
                   }
          }
\right)^{1/2}, \quad
\sin \theta = 
 \frac{1}{\sqrt{2}} \left( 1 -
 \frac{1}
         {\sqrt{      1 +     \frac{4k^2/m^2}
                          {(\omega_\ry^2 - \omega_\rx^2)^2}
                   }
          }
\right)^{1/2}. 
\label{eq:cos-theta-sin-theta}
\end{align}
where 
\[\frac{4k^2/m^2}{(\omega_\ry^2 - \omega_\rx^2)}
 = 4  \left(
 \frac{ \sqrt{m_\rx  m_\ry} }{k} 
 (\omega_{\ry, {\bf o}}^2 - 
   \omega_{\rx, {\bf o}}^2  ) 
   + 
      \frac{m_\rx - m_\ry}{\sqrt{m_\rx m_\ry} }
   \right)^{-2}.
 \]
Note that the rotation angle quantifies the strength of the coupling, in the sense that $\cos \theta \approx 1$ and $\sin \theta \approx 0$  when the coupling constant $k$ is small and the oscillators are off-resonant.
Define the normalization parameter $\mu := \sqrt[4]{\frac{m_\rx}{m_\ry}}$. Then we choose the transformed position and momenta operators to be given by
\begin{align}
( \hat{u} , \hat{v})          &:= {\bf R} (\mu^{-1}  \hat{x} , \mu \hat{y} ),       \notag\\
(\hat{p}_u , \hat{p}_v )      &:= {\bf R} (\mu \hat{p}_x ,\mu^{-1} \hat{p}_y )
\end{align}
where $(x_1,x_2,...)$ denotes a column vector.
The quantized Hamiltonian is
$\tH =  \tH_{ m,\w_u; \hat u , \hat p_u} + \tH_{m,\w_v; \hat v , \hat p_v}$
where
\begin{align}
 \omega_\ru = \sqrt{ \begin{pmatrix} \w^2_\rx \\ \w^2_\ry \end{pmatrix} \cdot \begin{pmatrix} \cos^2 \theta \\ \sin^2 \theta \end{pmatrix}  - \frac{k}{m} \sin(2 \theta) },
 \quad
 \omega_\rv = \sqrt{ \begin{pmatrix} \w^2_\rx \\ \w^2_\ry \end{pmatrix} \cdot \begin{pmatrix} \sin^2 \theta \\ \cos^2 \theta \end{pmatrix}  + \frac{k}{m} \sin(2 \theta) }.
\end{align}
Note that the frequencies $\w_\ru$ and $\w_\rv$ are real as
 long as the original oscillator frequencies $\w_{\rx,{\bf o}}$ and $\w_{\ry,{\bf o}}$ before
 renormalization are also real, because the renormalized frequencies $\omega_\rx$ and $\omega_\ry$ increase as the coupling strength $k$ increases.

The linear canonical transformation that decouples the pair of harmonic oscillators is not unique. We chose the transformation that gives the same mass $m$ for the decoupled oscillators, so that the only parameter different between them are the frequencies $\omega_u$ and $\omega_v$. Since the transformation we have performed is canonical, $\lbrack \hat{u},\hat{v} \rbrack = \lbrack \hat{p}_u,\hat{p}_v \rbrack = \lbrack \hat{u},\hat{p}_v \rbrack = \lbrack \hat{v},\hat{p}_u \rbrack = 0$, $\lbrack \hat{u},\hat{p}_u \rbrack = i\hbar$ and $\lbrack \hat{v},\hat{p}_v \rbrack = i\hbar$. Hence there exist Hilbert spaces $\cU, \cV$ isomorphic to $L^2(\bbR) $ such that $\cX \otimes \cY = \cU \otimes \cV$, $\hat u = \hat x_\cU \otimes \bbone_\cV$, $\hat v = \bbone_\cU \otimes \hat x_\cV $, $\hat p_u = \hat p_\cU \otimes  \bbone_\cV  $ and $\hat p_v = \bbone_\cU \otimes \hat p_\cV $.

Let $c_\ru := \sqrt{\frac{m \omega_\ru}{\hbar}}$ and $c_v := \sqrt{\frac{m \omega_v}{\hbar}}$, $c_\rx := \sqrt{\frac{m_\rx \omega_\rx}{\hbar}}$ and $c_\ry := \sqrt{\frac{m_\ry \omega_\ry}{\hbar}}$.  Then the set of eigenstates of the uncoupled Hamiltonian
$    \tH_{ m_\rx, \omega_{\rx} ; \hat x, \hat p_\rx  } + \tH_{ m_\ry, \omega_\ry ; \hat \ry, \hat p_\ry  } $
and the full Hamiltonian $\tH$ are
$\{ | \psi_{\kappa, c_\ru, \cU} , \psi_{\chi, c_\rv, \cV} \>  \} _{\kappa ,\chi \in \bbN}  $ and
$\{ | \psi_{j, c_\rx, \cX} , \psi_{\ell, c_\ry, \cY} \>  \} _{j,\ell \in \bbN}  $
 respectively.
\section{Truncated dynamics of the interacting system} \label{sec:model-description}
This section highlights the main results of our paper. We provide a computable approximation to our physical model's truncated channel with corresponding error bounds that are simple to describe.

\subsection{The general model}
The Hilbert space of our model has the general form $\cH = \cX \otimes \cY$ where $\cX$ and $\cY$ are separable Hilbert spaces of the system and the environment respectively. Our model's Hamiltonian is
 \[
 \tH := \tH_\rx \otimes \bbone_\cY + \bbone_\cX \otimes \tH_\ry + \tH_{{\rm int}}
 \]
 where $\tH, \tH_\rx, \tH_\ry$ and $ \tH_{{\rm int}}$ are (typically unbounded) Hermitian operators in the sets $L(\cH), L(\cX), L(\cY)$ and $L(\cH)$ respectively. The Hamiltonians $\tH_\rx$ and $\tH_\ry$ describes the bare dynamics on $\cX$ and $\cY$ respectively, and $\tH_{{\rm int}}$ describes the system-bath interaction.

Let the initial state of the entire model be $\rho_{{\rm all}} := \rho_0 \otimes \sigma_\ry$, where $\rho_0 \in \fD(\cX)$ and the bath state is
$  \sigma_\ry = \sum_{\ell \in \bbN} p_\ell |\ell_\cY\>\<\ell_\cY| \in \fD(\cY). $
 Let the time evolution operator of the entire model at time $t$ be the unitary operator $\tens U_t \in \fB(\cH)$. Then the time evolved state of system $\cX$ at time $t$ is
\begin{align}
\rho_t &   := \Phi_t(\rho_0 ) = \tr_{\cY} \tU_t \rho_{{\rm all}} \tU_t^\dagger \notag\\
        &   = \sum_{j,j',\ell, \ell' \in \bbN}
                    \<\ell'_\cY| \tU_t |j_\cY\>\<j_\cY| \rho_0 \otimes p_\ell |\ell_\cY\>\<\ell_\cY| j'_\cY \>\<j'_\cY| \tU_t^\dagger |\ell'_\cY\>
                    \notag\\
    &= \sum_{\ell,\ell' \in \bbN}\<\ell'_\cY | \tU_t |\ell_\cY\> p_\ell \rho_0 \<\ell_\cY| \tU_t^\dagger |\ell'_\cY\>.
\end{align}
Using (\ref{eq:transition-amplitude-def}), a feasible Kraus set and transition amplitudes for $\Phi_t$ are
\begin{align}
\fK_t &:= \left \{ \sqrt{p_\ell}\<\ell'_\cY| \tU_t |\ell_\cY\> : \ell,\ell' \in \bbN \right\} \label{eq:kraust}\\
T_{\fK_t}^{(a,b) \to (a',b')} &=
	\sum_{\ell,\ell' \in \bbN} p_\ell
		\< b'_\cX, \ell'_\cY  | \tU_t | b_\cX, \ell_\cY\>
		\< a_\cX, \ell_\cY  | \tU_t^\dagger | a'_\cX, \ell'_\cY\> \label{eq:transition-amplitudet}.
\end{align}
The transition amplitudes (\ref{eq:transition-amplitudet}) of the full quantum channel may be impossible to evaluate because of the inifinite summation that is required. In view of this, we can instead evaluate the truncated transition amplitudes by truncating the infinite summation. These truncated transition amplitudes are the transition amplitudes of the truncated quantum channel.
\subsection{Coupled harmonic oscillators} \label{subsec:coupled-HOs}
The approximate dynamics of coupled harmonic oscillators are still actively studied \cite{ChC07, ChC09}. In our paper, we use the model as described in Section \ref{subsubsec:quantized-model}. Let $z_u := e^{-i\omega_u t}$ and $z_v := e^{-i\omega_v t}$.
For $j \in \bbN$, define $|j_\rx\>:= |\psi_{j, c_\rx,\cX}\>, |j_\ry\>:= |\psi_{j, c_\ry,\cY}\>, |j_\ru\>:= |\psi_{j, c_\ru,\cU}\>$, and $|j_\rv\>:= |\psi_{j, c_\rv,\cV}\>$.
 Then the unitary operator $\tU_t$ has the spectral decomposition
\begin{align}
\tU_t = \sum_{\kappa , \chi \in \mathbb N} \sqrt{z_u z_v} z_u^\kappa z_v ^\chi | \kappa_\ru, \chi_\rv \>\< \kappa_\ru, \chi_\rv |. \label{eq:propagator}
\end{align}
Let $r = \exp(-\frac{\hbar \w_y}{ k_B T}) \in [0,1)$, where $k_B$ is the Boltzmann constant and $0 \le T< \infty$ is the effective temperature of the bath.  The state of the bath with a Boltzmannian distribution is
\begin{align}
\sigma_{\ry} = \sum_{\ell \in \bbN} r^\ell (1-r) | \ell_\ry \>\<\ell_\ry |. \label{eq:bath-state}
\end{align}
Thus $p_\ell = r^\ell (1-r)$ in equations (\ref{eq:kraust}) and (\ref{eq:transition-amplitudet}).
\subsubsection{Kraus operators and transition amplitudes} \label{subsubsec:kraus-and-choi}
Using (\ref{eq:kraust}), the matrix elements of our Kraus operator $\tK \in \fK_t$ indexed by $\ell, \ell' \in \bbN$ are
\begin{align}
&\< j'_\rx  | \tK | j_\rx \>  =
\sqrt{p_\ell}  \<  j'_\rx  ,  \ell'_\ry  | \tU_t | j'_\rx  , \ell_\ry \>
\notag \\
 =&
\sqrt{p_\ell}  \< j'_\rx  ,  \ell'_\ry |\sum_{\kappa, \chi \in \mathbb N} z_\ru^{\kappa + \half}  z_\rv ^{\chi+\half}
 |\kappa_\ru, \chi_\rv \> \< \kappa_\ru, \chi_\rv | j_\rx, \ell_\ry \>
\label{eq:kraus-form}
\end{align}
The goal is now to find an expression for the truncated transition amplitudes for small values of $a,b,a'$ and $b'$. We first give an expression for the transition amplitude with respect to the full Kraus set $\fK_t$, which is
\begin{align}
T_{\fK_t}^{(a,b) \to (a',b')} =&
	\sum_{\ell,\ell' \in \bbN} r^\ell (1-r)
		\< b'_\rx, \ell'_\ry  | \tU_t | b_\rx, \ell_\ry\>
		\< a_\rx, \ell_\ry  | \tU_t^\dagger | a'_\rx, \ell'_\ry\> \notag\\
=&
	\sum_{\ell,\ell' \in \bbN}  r^\ell (1-r)
		\< b'_\rx, \ell'_\ry  | \sum_{\kappa ,\chi \in \bbN} z_u^{\kappa+\half} z_v^{\chi + \half} |\kappa_\ru, \chi_\rv\>
        \< \kappa_\ru , \chi_\rv | b_\rx, \ell_\ry\>\notag\\
&\quad \times		\< a_\rx, \ell_\ry  | \sum_{\kappa', \chi' \in \bbN}  z_u^{-\kappa'-\half} z_v^{-\chi'-\half}| \kappa'_\ru, \chi'_\rv\>
        \< \kappa'_\ru, \chi'_\rv | a'_\rx, \ell'_\ry\>. \label{eq:unsimplified-transition-amplitude}
\end{align}
The matrix elements in the expression above can be simplified by expressing them in the $x$ and $y$ coordinates of the original oscillators. In particular, the expression above becomes an integral of the product of rescaled Hermite functions.
To simplify notation, let $u_{x,y}:=  c_\ru ( \frac{x}{\mu c_\rx} \cos \theta +  \frac{\mu y}{c_\ry} \sin \theta  )$ and $v_{x,y}:=  c_\rv (- \frac{x}{\mu c_\rx} \sin \theta +  \frac{\mu y}{c_\ry} \cos \theta )$ denote the coordinates of the decoupled oscillators in the basis of the original oscillators. By making appropriate substitutions, we have that
\begin{align}
u_{x,y} &=  \sqrt{ \frac{m_\ry \omega_\ru  }{m_\rx \omega_\rx  }  } x \cos \theta
        + \sqrt{ \frac{m_\rx \omega_\ru  }{m_\ry  \omega_\ry }  } y \sin \theta \notag\\
v_{x,y} &=  -\sqrt{ \frac{m_\ry \omega_\rv }{m_\rx \omega_\rx }  } x \sin \theta
        + \sqrt{ \frac{m_\rx \omega_\rv }{m_\ry \omega_\ry }  } y \cos \theta.\label{eq:uxy_vxy}
\end{align}
The summation indices in the transition amplitude corresponding to the full Kraus set $\fK_t$ in (\ref{eq:unsimplified-transition-amplitude}) are $\ell$ and $\ell'$ respectively. In this paper, we choose our truncated Kraus set to be $\Omega_{L,t}$, where only the summation over $\ell'$ is truncated.

By applying Mehler's formula (\ref{eq:mehler2}) on the variable $\ell$, the expression for the truncated transition amplitude is
\begin{align}
 T_{ \Omega_{L,t}}^{(a,b) \to (a',b')} =& \sum_{\ell' \le L} \sum_{\substack{\kappa ,\chi \in \bbN \\ \kappa',\chi' \in \bbN} } z_u^{\kappa - \kappa'} z_v^{\chi-\chi'}
        f[a,b,a',b' ; \ell', \kappa, \kappa', \chi, \chi'] 
        \frac{\omega_\ru \omega_\rv}{\omega_\rx \omega_\ry}
      \sqrt{\frac{1-r}{\pi(1+r)} }
        \label{eq:truncated-transition-amplitude},
\end{align}
where the path-dependent and time-independent transition amplitudes $f[a,b,a',b' ; \ell', \kappa, \kappa', \chi, \chi']$ are
\begin{align}
&f[a,b,a',b' ; \ell', \kappa, \kappa', \chi, \chi'] :=
        \int_{{\bf x,y} \in \bbR^4} d {\bf x}\ d{\bf y}
        \exp\left[{-\frac{1+r^2}{2(1-r^2)}\Bigl(y_3^2 - \frac{4r  y_3y_4}{1+r^2} +y_4^2 \Bigr) } \right]
        \notag\\
&		
\quad \times    \Bigl ( \psi_{a'}(x_1) \psi_{b'}(x_2) \psi_{b}(x_3) \psi_{a}(x_4) \Bigr)
                    \Bigl ( \psi_{\ell'}(y_1) \psi_{\ell'}(y_2)  \Bigr)
\notag\\
&\quad \quad \times
        \psi_{\kappa'} (u_{x_1,y_1}) \psi_{\chi'} (v_{x_1,y_1})
        \psi_{\kappa}  (u_{x_2,y_2}) \psi_{\chi}  (v_{x_2,y_2}) \notag\\
& \quad \quad \quad  \times
        \psi_{\kappa}  (u_{x_3,y_3}) \psi_{\chi}  (v_{x_3,y_3})
        \psi_{\kappa'} (u_{x_4,y_4}) \psi_{\chi'} (v_{x_4,y_4}) \label{eq:transition-amplitude-main},
\end{align}
where ${\bf x} = (x_1,x_2,x_3,x_4)$ and ${\bf y} = (y_1,y_2,y_3,y_4)$ are the rescaled position coordinates of the system and environment oscillator respectively. The integral $f$ can be more easily evaluated if we express it as a product of three integrals, in the sense that
\begin{align}
&f[a,b,a',b' ; \ell', \kappa, \kappa', \chi, \chi'] =
        I_{a',\ell',\kappa', \chi'}
        I_{b',\ell',\kappa, \chi}
        J_{b,a,\kappa , \chi ,\kappa', \chi' , r}
        \label{eq:transition-amplitude-main2},
\end{align}
where the integrals are
\begin{align}
I_{a',\ell',\kappa',\chi'}
    &:= \int_{\bbR^2}
             \Theta_{x,y}(a' ; \ell')
              \psi_{\kappa'} (u_{x,y})  \psi_{\chi'} (v_{x,y})
        \ dx \ dy \label{eq:Ifunction-definition}\\
J_{b,a,\kappa ,\chi, \kappa', \chi',r}
    &:= \int_{\bbR^4}
             \Theta_{w,x,y,z}(b,a ; r)
            \psi_{\kappa} (u_{w,y})  \psi_{\chi} (v_{w,y})
            \psi_{\kappa'} (u_{x,z})  \psi_{\chi'} (v_{x,z})
        \ dw \ dx \ dy \ dz. \label{eq:Jfunction-definition},
\end{align}
and the kernels are
\begin{align}
\Theta_{x,y} (i ; j ) &:= \psi_i(x) \psi_j (y)    \\
\Theta_{w,x,y,z} (i ,j ; r ) &:= \psi_i(w) \psi_j (x)
      \exp\left[
       -\frac{1+r^2}{2(1-r^2)} \left(y^2 - \frac{4r}{1+r^2} yz + z^2 \right) \right].
\end{align}

Now our truncated transition amplitude is still a sum over an infinite number of integrals, and hence we intend to approximate it by the finite sum
 \begin{align}
 A_{ L,N,t }^{(a,b) \to (a',b')} =
 \sum_{\ell' \le L} 
 \sum_{\substack{
 					\kappa ,\chi \le N  \\ 
 					\kappa',\chi' \le N } 
	     } 
	z_u^{\kappa - \kappa'} z_v^{\chi-\chi'}
      f[a,b,a',b' ; \ell', \kappa, \kappa', \chi, \chi'] 
        \frac{\omega_\ru \omega_\rv}{\omega_\rx \omega_\ry}
       \sqrt{\frac{1-r}{\pi(1+r)} }
\label{eq:TA-approximation}
\end{align}
for some positive integer $N$.
In Theorem \ref{thm:qho-error-bound}, we prove that the absolute value of the error term by approximating the truncated transition amplitude (\ref{eq:truncated-transition-amplitude}) with (\ref{eq:TA-approximation}) vanishes as $N$ becomes large while $a,b,a',b'$ and $L$ remain small. The proof of our theorem uses mainly the Cauchy-Schwarz inequality, Lemma \ref{lem:integral-convergence1} and Lemma \ref{lem:integral-convergence2}. The intuition behind our technical lemmas are elementary consequences of the behavior of order $n$ Hermite functions in the oscillatory interval $[-\sqrt{n}, \sqrt{n}]$ and outside of it. Within the oscillatory interval, Hermite functions have amplitudes that vanish as $n$ gets large. Outside of the oscillatory region, Hermite functions have exponentially small amplitudes as their arguments becomes large. Thus we construct our upper bounds for the integral of the product of Hermite functions by performing the integration separately in two overlapping regions, as depicted in Figure \ref{fig:integration-region}.
%
\begin{theorem}\label{thm:qho-error-bound}
Let $m_\rx, m_\ry , \omega_\rx, \omega_\ry, t>0, k\ge 0 $ and $0 \le r < 1$ be real numbers, and $D,L$ and $N$ be positive integers. 
Let $C$ be a constant that depends on $m_\rx, m_\ry, \omega_\rx, \omega_\ry$ and $k$ (see (\ref{eq:constant-C})), and $A,\tilde A, B, \tilde B$ be constants that depend on $D$ and $L$ (see (\ref{eq:constant-AB})).
Let $ T_{ \Omega_{L,t}}^{(a,b) \to (a',b')}$ be the truncated transition amplitude defined in (\ref{eq:truncated-transition-amplitude}) have approximation $ A_{ L,N,t }^{(a,b) \to (a',b')}  $ given by (\ref{eq:TA-approximation}).
Then for all integers $0 \le a, b, a', b' \le D$,
\begin{align}
&\left|T_{ \Omega_{L,t}}^{(a,b) \to (a',b')} - 
 A_{ L,N,t }^{(a,b) \to (a',b')} \right|  \notag\\
\le & \left(
   \frac{4\sqrt{A^2 B}}{9(N-\frac{1}{2})^3} 
+\frac{4 \sqrt{A \tilde A B}}{3 N^{5/4} (N-\frac{1}{2})^{3/2}}
   \frac{e^{-N C/2}}{1-e^{-C/2} }
+ \frac{\sqrt{\tilde A^2 B}}{N^{5/2}} \frac{e^{-NC}}{(1-e^{-C/2})^2} +
\frac{ \sqrt{A^2 \tilde B}}{N^{5/2}} \frac{e^{-NC}}{(1-e^{-C/2})^2}
\right. \notag\\
 &\left.
+ \frac{2  \sqrt{A \tilde A \tilde B}}{N^{5/4} } 
   \frac{e^{-3 N C/2}}{(1-e^{-C/2})(1-e^{-C}) }
+\sqrt{ \tilde A^2 \tilde B} \frac{e^{-2NC}}{(1-e^{-C})^2} \right)^2  
        \frac{\omega_\ru \omega_\rv}{\omega_\rx \omega_\ry}
         \sqrt{\frac{1-r}{\pi(1+r)} }
          (L+1)
.
 \end{align}
\end{theorem}
\begin{remark}\label{remark:qho}
Observe that when $C$ is a very large number, the upper bound of the above theorem is dominated by the expression
\[
\frac{4 A^2 B  }{81 (N-\half)^{6}  }
\left( 
        \frac{\omega_\ru \omega_\rv}{\omega_\rx \omega_\ry}
         \sqrt{\frac{1-r}{\pi(1+r)} }
          (L+1) \right)
         \le \frac{14.7103 (L+1)  
         \tilde n_D^4 \tilde n_L^2 
         }{ (N-\half)^{6}  }
\left( 
        \frac{\omega_\ru \omega_\rv}{\omega_\rx \omega_\ry}
         \sqrt{\frac{1-r}{ 1+r } }
           \right)
\]
as $N$ becomes large (and $\tilde n_i$ is defined as  $\left( \max_{0\le j \le i} \|  \psi_j \|_1 \right)$.
\end{remark}
\begin{remark}\label{remark:parity}
If the parity of $a+b$ differs from that of $a'+b'$, the approximate truncated amplitude is necessarily identically zero for all positive integers $L$ and $N$. This is a result of a simple parity counting argument after noting that the I-type integrals (\ref{eq:Ifunction-definition}) and J-type integrals (\ref{eq:Jfunction-definition}) are zero whenever the parity of the sum of their indices are odd. Hence parity is conserved with regards to our physical model.
\end{remark}
\begin{remark}\label{remark:tighter-bound}
The bounds of Theorem \ref{thm:qho-error-bound} can be substantially tightened using information pertaining to the I-type integrals (\ref{eq:Ifunction-definition}), which are substantially simplier to evaluate than the J-type integrals (\ref{eq:Jfunction-definition}). Using bounds for the J-type integrals (\ref{ineq:J-upper-bound}), we have that for integer $N'$ greater than $N$, 
\begin{align}
\left|T_{ \Omega_{L,t}}^{(a,b) \to (a',b')} - 
 A_{ L,N,t }^{(a,b) \to (a',b')} \right| 
		\le  &
\sum_{0 \le \ell \le L } 
\sum_{ 
\substack{  N < \kappa, \kappa' \le N' \\
			 N < \chi, \chi' \le N' }
} 
|I_{a',\ell',\kappa',\chi'} 
I_{b',\ell', \kappa,\chi}|
 B (\kappa \kappa' \chi \chi')^{-5/2} \notag\\
 &+ 
 \left|
 |T_{ \Omega_{L,t}}^{(a,b) \to (a',b')} - 
 A_{ L,N',t }^{(a,b) \to (a',b')} \right| .
\end{align}
\end{remark}
\begin{proof}[Proof of Theorem \ref{thm:qho-error-bound}]
Our goal is to obtain upper bounds on each of the integrals $I$ and $J$ defined in (\ref{eq:Ifunction-definition}) and (\ref{eq:Jfunction-definition}). Applying the Cauchy-Schwarz inequality on $| I_{a',\ell',\kappa',\chi'} |$ gives
\begin{align}
| I_{a',\ell',\kappa',\chi'} | \le
        \sqrt{\int_{\bbR^2} \bigl|  \Theta_{x,y}(a' ; \ell')  \bigr |  \psi_{\kappa'} (u_{x,y})^2
        \ dx \ dy}
        \sqrt{\int_{\bbR^2} \bigl|   \Theta_{x,y}(a' ; \ell')  \bigr|  \psi_{\chi'} (v_{x,y})^2
        \ dx \ dy}.
\end{align}
We similarly use the Cauchy-Schwarz inequality to obtain an upper bound of the absolute value of (\ref{eq:Jfunction-definition}), which is
\begin{align}
|J_{b,a,\kappa ,\chi, \kappa', \chi',r}| \le&
        \sqrt{\int_{\bbR^2} \bigl|  \Theta_{w,x,y,z}(b,a ; r) \bigr |
            \psi_{\kappa} (u_{w,y})^2
            \psi_{\kappa'} (u_{x,z})   ^2
        \ dx \ dy} \notag\\
&\quad \times
        \sqrt{\int_{\bbR^2} \bigl|   \Theta_{w,x,y,z}(b,a ; r) \bigr|
              \psi_{\chi} (v_{w,y})^2
              \psi_{\chi'} (v_{x,z})^2
        \ dx \ dy}.
\end{align}
For the purpose of using Lemma \ref{lem:integral-convergence1} and Lemma \ref{lem:integral-convergence2}, define the constants
\begin{align}
c_1 &= \min \left\{
        \sqrt{ \frac{m_\ry \omega_\ru}{ m_\rx \omega_\rx }  } \cos \theta,
        \sqrt{ \frac{m_\rx \omega_\ru}{ m_\ry \omega_\ry }  } \sin \theta
        \right\}\notag \\
c_2 &= \min \left\{
        \sqrt{ \frac{m_\ry \omega_\rv}{ m_\rx \omega_\rx }  } \sin \theta,
        \sqrt{ \frac{m_\rx \omega_\rv}{ m_\ry \omega_\ry }  } \cos \theta
        \right\} \notag \\ 
        C &= \min\{ 1/(4c_1^2), 1/(4c_2^2) \}   \label{eq:constant-C}
 \end{align}
 and 
 \begin{align}
A &= 1.74^2 
     \left( \max_{0\le j \le D} \|  \psi_j \|_1 \right)
     \left( \max_{0\le j \le L} \|  \psi_j \|_1 \right)
, \quad \tilde A = (4.74)(2e^{-\half})^{D+L}\sqrt{D! D^D L! L^L} \notag\\
B &= 57.6 
     \left( \max_{0\le j \le D} \|  \psi_j \|_1 \right)^2
, \quad \tilde B = (39.6)(2e^{-1})^{D} D! D^D  
\label{eq:constant-AB}.
\end{align}
Noting that $\sin \theta$ and $\cos \theta$ are positive by definition (see (\ref{eq:cos-theta-sin-theta})), and using the definitions of $c_1, c_2$ and $C$ with Lemma \ref{lem:integral-convergence1}, we have that
\begin{align}
| I_{a',\ell',\kappa',\chi'} |&  \le
\sqrt{
        A (\kappa')^{-5/2}
        + \tilde A
         e^{- \kappa' C }
             }
\sqrt{
         A (\chi')^{-5/2}
        + \tilde A
         e^{- \chi'  C  }
             } \label{ineq:I-upper-bound}
\end{align}
and we have a similar upper bound of $|I_{b',\ell',\kappa, \chi}|$. Using Lemma \ref{lem:integral-convergence2}, we have that
\begin{align}
|J_{b,a,\kappa ,\chi, \kappa', \chi',r}|
&\le
   \sqrt{
         B (\kappa \kappa')^{-5/2}
      + \tilde B
        e^{-(\kappa + \kappa') C }
   }
   \sqrt{
        B (\chi \chi')^{-5/2}
      + \tilde B
        e^{-(\chi + \chi') C   }
   }. \label{ineq:J-upper-bound}
\end{align}
By expanding out the terms of the products of the upper bounds given by (\ref{ineq:I-upper-bound}) and (\ref{ineq:J-upper-bound}), our upper bound on the absolute value of  (\ref{eq:transition-amplitude-main2})
\[
\Bigl |
    f[a,b,a',b' ; \ell', \kappa, \kappa', \chi, \chi']
 \Bigr | 
 \le \sqrt{
W_{\kappa,\kappa'} W_{\chi,\chi'}
}
\]
where 
\begin{align}
W_{\kappa, \kappa'} :=& 
    \frac{A^2 B}{\kappa^5 (\kappa')^5}
+ \frac{A \tilde A B}{\kappa^{5} (\kappa')^{5/2} }  e^{-\kappa' C}
+ \frac{A \tilde A B}{\kappa^{5/2} (\kappa')^{5} } e^{-\kappa C}
+ \frac{\tilde A^2 B}{\kappa^{5} (\kappa')^{5/2} }  e^{-(\kappa + \kappa' )C}
\notag\\ & 
+  \frac{A^2 \tilde B}{\kappa^{5/2} (\kappa')^{5/2} } e^{-(\kappa+\kappa' )C}
+ \frac{A \tilde A  \tilde B}{\kappa^{5/2}  } e^{-(\kappa+2\kappa' )C}
+ \frac{A \tilde A  \tilde B}{ (\kappa')^{5/2} } e^{-(2\kappa+\kappa' )C}
+ {\tilde A^2  \tilde B}  e^{-(2\kappa+2\kappa' )C}\notag
.
\end{align}
By the subadditivity of the square root function, we have that 
\begin{align}
\sqrt{W_{\kappa, \kappa'} } \le& 
    \frac{\sqrt{A^2 B}}{\kappa^{5/2} (\kappa')^{5/2}}
+ \frac{\sqrt{A \tilde A B}}{\kappa^{5/2} (\kappa')^{5/4} }  e^{-\kappa' C/2}
+ \frac{\sqrt{A \tilde A B}}{\kappa^{5/4} (\kappa')^{5/2} } e^{-\kappa C/2}
+ \frac{\sqrt{\tilde A^2 B}}{\kappa^{5/2} (\kappa')^{5/4} }  e^{-(\kappa + \kappa' )C/2}
\notag\\ & 
+  \frac{\sqrt{A^2 \tilde B}}{\kappa^{5/4} (\kappa')^{5/4} } e^{-(\kappa+\kappa' )C/2}
+ \frac{\sqrt{A \tilde A  \tilde B}}{\kappa^{5/4}  } e^{-(\kappa+2\kappa' )C/2}
+ \frac{\sqrt{A \tilde A  \tilde B}}{ (\kappa')^{5/4} } e^{-(2\kappa+\kappa' )C/2}
+ \sqrt{\tilde A^2  \tilde B}  e^{-(\kappa+\kappa' )C}\notag
.
\end{align}
The summation of the above expression over $\kappa$ and $\kappa'$ can be seen as an inner product of vectors with exponentially decaying terms and polynomially decaying terms respectively. We hence apply H\"{o}lder's inequality for sequence spaces on the summation of the above expression over $\kappa$ and $\kappa'$. We thereby obtain an upper bound of the sum in terms of the one-norm of the vector with exponentially decaying terms, and the infinity-norm of the vector with polynomially decaying terms:
\begin{align} 
   \sum_{\kappa, \kappa' \ge N}
   \sqrt{W_{\kappa,\kappa'}}
 \le&
   \sqrt{A^2 B} \Bigl(\sum_{X \ge N} X^{-5/2} \Bigr)^2
+\frac{2 \sqrt{A \tilde A B}}{N^{5/4} } \Bigl(\sum_{X \ge N} X^{-5/2} \Bigr) 
   \frac{e^{-N C/2}}{1-e^{-C/2} }
+ \frac{\sqrt{\tilde A^2 B}}{N^{5/2}} \frac{e^{-NC}}{(1-e^{-C/2})^2} \notag\\
&
+ 
  \frac{ \sqrt{A^2 \tilde B}}{N^{5/2}} \frac{e^{-NC}}{(1-e^{-C/2})^2}
+\frac{2  \sqrt{A \tilde A \tilde B}}{N^{5/4} } 
   \frac{e^{-3 N C/2}}{(1-e^{-C/2})(1-e^{-C}) }
+\sqrt{ \tilde A^2 \tilde B} \frac{e^{-2NC}}{(1-e^{-C})^2}
. \notag
\end{align}
Now the integral $\int_{N-\frac{1}{2}}^\infty x^{-5/2} dx$ is an upper bound of the sum $\sum_{X \ge N}^\infty X^{-5/2}$ by the convexity of the integrand. Hence $\sum_{X \ge N}^\infty X^{-5/2} \le \frac{2}{3}(N-\frac{1}{2})^{-3/2}$, and we can subsitute this bound into our upper bound of the square of $\sum_{\kappa, \kappa' \ge N}W_{\kappa,\kappa'}$ summed over the index $0 \le \ell' \le L$ to get the result.
\end{proof}
\section{Bounds on Hermite functions}\label{sec:hermite-function-bounds}
This section provides the main technical lemmas that are used to obtain error bounds on our approximation to our truncated transition amplitudes.
The main technical tools that we use in this section are Alzer's sharp bounds on the gamma function \cite{Alz09} and bounds on the Dominici's asymptotic approximation of Hermite functions with error estimates by Kerman, Huang and Brannan \cite{KHB09}.
\begin{lemma} \label{lem:hermite-oscillatory}
For all positive integers $n$ and reals $ x \in [-\sqrt{n}, \sqrt{n}]$, we have
$ |\psi_n(x)|  < 1.74 n^{-5/4}$.
\end{lemma}
\begin{proof}
This proof combines Alzer's sharp bounds on the gamma function \cite{Alz09} with uniform bounds on the envelope of the Hermite functions in the oscillatory region by Kerman, Huang and Brannan \cite{KHB09}. Using Kerman, Huang and Brannan's result (see equation (2.1) and (1.4) in \cite{KHB09}), for $x \in [-\sqrt{n}, \sqrt{n}]$ we have
\begin{align*}
|\psi_n(x)| \le 2^{-3/4}\sqrt{35} \frac{\sqrt{n!}\pi^{-1/4}2^{-n/2}}{\sqrt{2} \Gamma((n/2)+1)} n^{-1}.
\end{align*}
Alzer's sharp bounds for the gamma function are that for all $n > 0$,
\begin{align*}
 1<
\frac{  \Gamma(n+1) }{\sqrt{2 \pi n} \left(\frac{n}{e} \right)^n \left(n \sinh \frac{1}{n} \right)^{n/2} }
  < 1+\frac{1}{1620}{n^5} .
\end{align*}
Note that for real $n\ge 1 $, we have $1 \le \left(n \sinh \frac{1}{n} \right)^{n/2} < 1.085$. Hence we have that for positive integers $n$,
\begin{align*}
\frac{\sqrt{\Gamma(n+1)}  }{\Gamma(\frac{n}{2} + 1)}
    < \frac  { \sqrt{\sqrt{2 \pi n}            \left(\frac{n}{e} \right)^n  (1.085)(1+\frac{1}{1620 n^5} )}}
             { \sqrt{\pi n }    \left(\frac{n}{2e} \right)^{n/2} }
    <  0.9308 (2^{n/2}   ) n^{-1/4}.
\end{align*}
Hence Kerman, Huang and Brannan's upper bound on the envelope of the Hermite function for $x \in [-\sqrt{n} , \sqrt{n}]$ and $n \ge 1$ becomes
\begin{align*}
|\psi_n(x)| \le     2^{-3/4}\sqrt{35} \frac{\pi^{-1/4}2^{-n/2}}{\sqrt{2} }
                        \frac{\sqrt{ \Gamma(n+1) } }{\Gamma(\frac{n}{2}+1)} n^{-1}
            <      1.74 n^{-5/4}.
\end{align*}
\end{proof}
The next lemma provides a rather coarse upper bound on the absolute value of the Hermite function, with maximum utility in the monotonic region of the Hermite function.
\begin{lemma} \label{lem:hermite-monotonic}
For all reals $|x| > 1$ and integers $n \ge 0$,
we have
$|\psi_n(x)|
        \le 2^n \sqrt{\frac{n! n^n}{e^n \sqrt{\pi}}} e^{-x^2/4}$.
\end{lemma}
\begin{proof}
Using the Maclaurin decomposition of the Hermite polynomial $H_n(x)$ \cite{erdelyi}, we get
\begin{align*}
|\psi_n(x)| &\le n! n (2|x|)^n \frac{e^{-x^2/2}}{\sqrt{2^n n! \sqrt{\pi}}}
            = \sqrt{\frac{2^n n!}{\sqrt{\pi}}} n |x|^{n} e^{-x^2/2}.
\end{align*}
It is easy to verify that
$\displaystyle \sup_{x \in \bbR}\{ |x|^{n} e^{- x^2 / 4} \} = \left(\frac{2n }{e}\right)^{n/2}.$
Hence when $|x| > 1,$
\begin{align*}
&|\psi_n(x)| \le  2^n  \sqrt{\frac{n! n^n}{e^n \sqrt{\pi}}} e^{-x^2/4}. 
\end{align*}
\end{proof}
Lemma \ref{lem:integral-convergence1} provides upper bounds on the one-norm of the product of Hermite functions in terms of the order of the Hermite functions, and is crucial in obtaining upper bounds on our error estimates. 
\begin{lemma} \label{lem:integral-convergence1}
For real numbers $a,b \neq 0$, let $c = \min(|a|,|b|)$. Let $j,k,n \in \bbN$ and $n \ge 1$. Then
\[
\int_{\bbR^2} \psi_n(ax + by)^{2} | \psi_j(x) \psi_k(y) | \ dx dy
     \le
		 \frac{(1.74^2)\| \psi_j\|_1 \|\psi_k\|_1}{ n^{5/2}}
        + (4.74)
             2^{j+k} \sqrt{\frac{j!k! j^j k^k}{e^{j+k}}}
         e^{- n / (8c^2) }.
\]
\end{lemma}
\begin{proof}
We split the region over which the integral is performed into two overlapping regions $A_1$ and $A_2$ (see Figure \ref{fig:integration-region}), where
\[
A_1 = \{ (x,y) : |ax|+|by|\le \sqrt{n}\}
\] and
\[
A_2 = \{(x,y):  x^2+y^2 > \frac{\sqrt{n}}{\sqrt{2}c} \}.
\]
\begin{figure}[htb]
  \centering
    \includegraphics[width=0.5\textwidth]{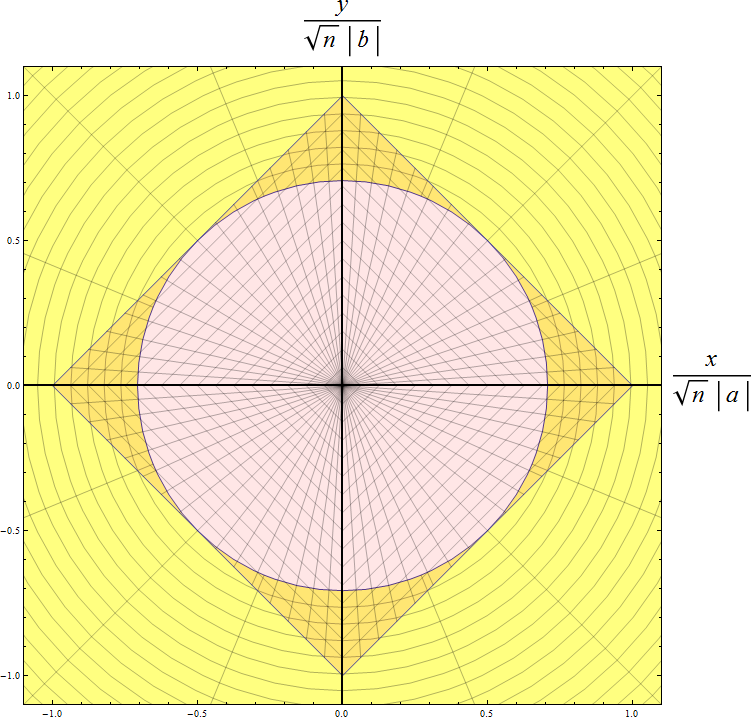}
      \caption{The area inside the square corresponds to the region $A_1$ (where $|\psi_n(ax+by)| \le 1.74n^{-5/4}$), and the area outside the circle corresponds to the region $A_2$ (where all Hermite functions decay exponentially).
      }\label{fig:integration-region}
\end{figure}
Then using the Charlier-Cram\'{e}r bound \cite{erdelyi} which states that $\sup_{n,x} |\psi_n(x)| \le \frac{1.086435}{\pi^{1/4}}$, and the uniform upper bound of the envelope of the Hermite polynomial in the oscillatory region as stated in Lemma \ref{lem:hermite-oscillatory}, we get
\begin{align*}
    &\int_{\bbR^2} \psi_n(ax + by)^{2} |\psi_j(x) \psi_k(y) | \ dx dy \\
\le & \int_{A_1} \frac{1.74^2}{n^{5/2}}
             |\psi_j(x) \psi_k(y) | \ dx dy
        + \int_{A_2} \frac{1.086435^2}{\sqrt{\pi}}  | \psi_j(x) \psi_k(y) | \ dx dy.
\end{align*}
The integral of $ |\psi_j(x) \psi_k(y) |$ over the region $A_1$ is at most
\begin{align*}
 \frac{(1.74^2)\| \psi_j\|_1 \|\psi_k\|_1}{ n^{5/2}}.
\end{align*}
Using the exponential upper bound of Lemma \ref{lem:hermite-monotonic}, the integral over the annulus region $A_2$ is at most
\begin{align*}
 &\frac{1.086435^2}{\sqrt{\pi}}
             2^{j+k} \sqrt{\frac{j!k! j^j k^k}{e^{j+k} \pi}}
        \int_{A_2}  e^{-(x^2+y^2)/4} \ dx dy
        \\
\le   &
\frac{1.086435^2}{\sqrt{\pi}}
             2^{j+k} \sqrt{\frac{j!k! j^j k^k}{e^{j+k} \pi }}
       2\pi   \int_{r > \frac{\sqrt{n}}{\sqrt{2} c}  } r e^{-r^2/4} \ dr \\
=   &
   1.086435^2
             2^{j+k+1} \sqrt{\frac{j!k! j^j k^k}{e^{j+k}}}
        (2 e^{- n / (8c^2) }). 
\end{align*}
Combining the upper bounds for region $A_1$ and $A_2$ then gives the result.
\end{proof}
The following lemma is needed to obtain upper bounds on the absolute value of $|J_{b,a,\kappa ,\chi, \kappa', \chi',r}| $, and is similar to the preceding lemma.
\begin{lemma} \label{lem:integral-convergence2}
Let $0 \neq a,b \in \bbR$, $c = \min(|a|,|b|)$ and $ j,k,n,n' \in \bbN$ where $n,n' \ge 1$. Then
\begin{align*}
&\int_{\bbR^2}
      \psi_{n}(a w + b y)^{2}
      \psi_{n'}(a x  + b z)^{2}
     |  \Theta_{w,x,y,z}(j , k ; r)  | \ dx dy\\ 
     \le &
\frac{57.6 \|\psi_j\|_1 \|\psi_k\|_1}{ (n n')^{5/2}}
      + (39.6)
        2^{j+k} \sqrt{\frac{j!k! j^j k^k}{e^{j+k} }}
        e^{-(n+n')/(8c^2)}.
\end{align*}
\end{lemma}
\begin{proof}
We split the region over which the integral is performed into two overlapping regions $A_1$ and $A_2$ (just as in the proof of Lemma \ref{lem:integral-convergence1}), where
\[
A_1 = \{ (w,x,y,z) : |a w|+|b y| \le \sqrt{n} , |a x|+|b z| \le \sqrt{n'}\}
\]
 and
\[
A_2 = \{(w,x,y,z):  w^2+y^2 > \frac{\sqrt{n}}{\sqrt{2} c} ,
                             x^2 + z^2 > \frac{\sqrt{n'}}{ \sqrt{2} c} \}.
\]
Then using the Charlier-Cram\'{e}r bound, Lemma \ref{lem:hermite-oscillatory}, and Lemma \ref{lem:hermite-monotonic}, we get
\begin{align*}
    &\int_{\bbR^2}
      \psi_{n}(a w + b y)^{2}
      \psi_{n'}(a x  + b z)^{2}
     |  \Theta_{w,x,y,z}(j , k ; r)  | \ dw\ dx\ dy\ dz \\
\le & \int_{A_1}
\frac{1.74^4}{(n n')^{5/2}}
    |  \Theta_{w,x,y,z}(j , k ; r)  |  \ dw\ dx\ dy\ dz
        + \int_{A_2}  \frac{1.086435^4}{\pi}
        |  \Theta_{w,x,y,z}(j , k ; r)   | \ dw\ dx\ dy\ dz.
\end{align*}
Now observe that we can express the kernel $ \Theta_{w,x,y,z}(j , k ; r) $ as
\begin{align*}
 \Theta_{w,x,y,z}(j , k ; r) =
e^{- \half (y^2 + z^2)}
\exp\left[
           -\frac{r}{1-r^2} (y-z)^2
\right]
\psi_j(w)
\psi_k(x)
\end{align*}
and hence the absolute value of our kernel has an upper bound that factorizes, in the sense that 
\begin{align*}
| \Theta_{w,x,y,z}(j , k ; r) | \le
e^{-\half (y^2 + z^2)}
|\psi_j(w) \psi_k(x)|.
\end{align*}
Hence the integral over the region $A_1$ is at most 
\[
\frac{1.74^4 (2\pi)\|\psi_j\|_1 \|\psi_k\|_1}{ (n n')^{5/2}}.
\]
The upper bound on the kernel also allows us to find that the integral over the region $A_2$ is at most
\begin{align*}
&   \frac{1.086435^4}{\pi}
        2^{j+k} \sqrt{\frac{j!k! j^j k^k}{e^{j+k} \pi}}
        \int_{A_2}  e^{-(w^2+x^2)/4}  e^{-\half (y^2 + z^2)}
        \ dw\ dx \ dy\ dz \\
\le & \   \frac{1.086435^4}{\pi^{3/2}}
        2^{j+k} \sqrt{\frac{j!k! j^j k^k}{e^{j+k} }}
        (2\pi)^2 \int_{r_1 \ge \frac{\sqrt{n}}{\sqrt{2} c}}
              e^{-r_1^2/4} dr_1
        \int_{r_2 \ge \frac{\sqrt{n'}}{ \sqrt{2} c}}
              e^{-r_2^2 /4}
        dr_2\\
\le &    \frac{1.086435^4(2\pi)^2}{\pi^{3/2}}
        2^{j+k} \sqrt{\frac{j!k! j^j k^k}{e^{j+k} }}
        4e^{-(n+n')/(8c^2)}\\
 \le &  (39.6)
        2^{j+k} \sqrt{\frac{j!k! j^j k^k}{e^{j+k} }}
        e^{-(n+n')/(8c^2)}.
\end{align*}
Combining our upper bounds on the integrals in the regions $A_1$ and $A_2$ thereby gives the result.
\end{proof}
\section{A case study: Off-resonant weakly coupled oscillators} \label{sec:example}
In this section, we perform a case study of the dynamics of our physical model when the oscillators are off-resonant and weakly coupled. In agreement with the standard results of perturbation theory, we find negligible amplitude damping in the truncated dynamics of our system. Moreover, we show that leakage error is the dominant error process, and can be mitigated via quantum error correction.

\subsection{Parameters of the physical model and the truncated channel}
The amount of truncation in our truncated channel $\cN$ is quantified by the parameter $L=2$, and the order of our approximation to the truncated channel is quantified by the parameter $N=6$ (see (\ref{eq:TA-approximation}) for the definition of the truncated transition amplitudes). We restrict the analysis of our truncated channels to a 4-level system by setting the parameter $D=3$.

Table \ref{table:qho-parameters} shows the parameters used for our physical model. We use the SI units. The output parameters are numerically computed using floating point numbers with 1024 bits of precision, and are shown up to ten decimal places.
\begin{table}[ht]
\centering
\begin{tabular}{|c| c| c| c|c|c|}
\hline
 &  Input parameters  &  &   Output parameters   & & Other parameters and bounds  \\
\hline
$m_{\rx}$ & $10^{-6} $ (kg)& 
$\w_{\rx}$& 1000499.8750624610 (Hz)&
$\|\psi_0\|_1$  & $\le 1.8827925275534299$\\
$m_{\ry}$ & $2\times 10^{-6}$ (kg)   &  
$\w_{\ry}$& 10000024.9999687501 (Hz)&
$\|\psi_1\|_1$ & $\le 2.124503864054394$ \\
$\omega_{\rx, {\bf o }} $ & $10^{6}$ (Hz) &
$\w_{\ru}$& 1001004.5438332595(Hz)&
$\|\psi_2\|_1$ & $\le 2.285324224284738$ \\
$\omega_{\ry , {\bf o}}$ & $10^{7}$ (Hz) & 
$\w_{\rv}$& 10000025.0001779494 (Hz)&
$\|\psi_3\|_1$ & $\le 2.410237758997186$\\
k & 100  & 
 $u_1$ & 10000025.0001779494
&  &  \\
 r & 0  &   
 $u_2$ & 0.0000031958 
 & &\\
  &   &  
 $v_1$ & -0.0000451621 
 & &\\
  &    & 
 $v_2$ & 1.0000000000 
 & &\\
\hline
\end{tabular}
\caption{We tabulate the important parameters of our physical model. Here, $u_1, u_2,v_1 $ and $v_2$ are defined implicitly in the equations $u_{x,y} = u_1 x + u_2 y$ and $v_{x,y} = v_1 x + v_2 y$ (see (\ref{eq:uxy_vxy})). Also note that $\|\psi_n\|_1 := \int_\bbR |\psi_n(x)| dx$. \label{table:qho-parameters}
}
\end{table}
For our choice of parameters, the positive constants $
A \tilde A B, 
\tilde A^2 B, 
A^2 \tilde B, 
A \tilde A \tilde B,
 \tilde A^2 \tilde  B 
 \le 10^{-288},$
and are negligible, and hence Remark \ref{remark:qho} holds. 
\subsection{Approximate dynamics of the truncated channel}
We numerically evaluate approximate truncated transition amplitudes (\ref{eq:TA-approximation}) corresponding to the transitions within the lowest energy levels of the system. The evaluation of each such approximate truncated transition amplitude $A_{2,6,t}^{(a,b)\to (a',b')}$ is a sum of 7203 terms for our choice of $L=2$ and $N=6$. We obtain the corresponding error bounds of our approximation from Remark \ref{remark:tighter-bound} with $N' = 15$. The error of each of our approximate truncated transition amplitude is at most 0.00084 -- a negligible amount. We plot magnitudes of the non-negligible truncated transition amplitudes in Figure \ref{fig:truncated-dynamics}.

Numerically, we find that the non-negligible terms of
$A_{2,6,t}^{(a,b)\to (a',b')}$ have values of $0 \le a,b \le 1$ and $0 \le a',b'\le 3$ that satisfy the relation
\begin{align}
\frac{a' - a}{2}  , \frac{b' - b}{2} \in \bbZ.
\end{align}
The above relation holds for two reasons. Firstly, the conservation of the parity of all truncated transition amplitudes as stated in Remark \ref{remark:parity} implies that the parity of $a+b$ equals the parity of $a'+b'$. Secondly, the other negligible transitions are in agreement with the results of using perturbation theory on off-resonant and weakly coupled harmonic oscillators. Also note that the negligible damping from the first excited state to the ground state of our truncated channel suggests that the coupled oscillator model is inconsistent with the phenomenon of amplitude damping \cite{nielsen-chuang} even at zero temperature.

The Choi-Jamiolkowski (CJ) operator of the finite input and output dimension channel $\Phi$ with Kraus set $\fK$ is the linear operator 
\begin{align}
\chi_{\Phi} := \sum_{\tK \in \fK} | \tK \rrangle \llangle \tK| =  \sum_{j,j'} \sum_{k,k'}
 \left( \sum_{\tK \in \fK } \tK_{j,j'}\tK_{k,k'}^*  \right)
 |j_\cK,  j'_\cH\>\<  k_\cK ,k'_\cH|.
\end{align}
where the linear map 
$| \cdot \rrangle : L(\cH, \cK) \to \cK \otimes \cH$ is a stacking isomorphism such that 
\begin{align*}
| \sum_{i,j} a_{i,j} |i_\cK\>\<j_\cH| \rrangle := \sum_{i,j} a_{i,j} |i_\cK\> |j_\cH\>.
\end{align*}

We now construct a quantum operation $\cA$ to approximate the truncated channel $\cN$.
  Abbreviating our approximate truncated transition amplitudes as $A_{aba'b'}:= A_{2,6,t}^{(a,b)\to(a',b')}$, our approximation to the CJ operator of our truncated channel is
\begin{align}
\chi :=
\begin{pmatrix}
 A_{0000} & A_{0002} & A_{0101} & A_{0103} \\
 A_{0020} & A_{0022} & A_{1012} & A_{0123} \\
 A_{1010} & A_{1012} & A_{1111} & A_{1113} \\
 A_{1030} & A_{1032} & A_{1131} & A_{1133}
\end{pmatrix}
\end{align}
where the bases labeled by the rows and columns are $|0,0\>, |2,0\>, |1,1\>, |3,1\>$ and $\<0,0|, \<2,0|, \<1,1|, \<3,1|$ respectively. The first and second entries of our bras and kets correspond to the output and input Hilbert spaces of the truncated map respectively. 
We plot the absolute value of some of these non-negligible matrix elements in Figure \ref{fig:truncated-dynamics}. 
\begin{figure}[htb]
  \centering
    \includegraphics[width=1.0\textwidth]{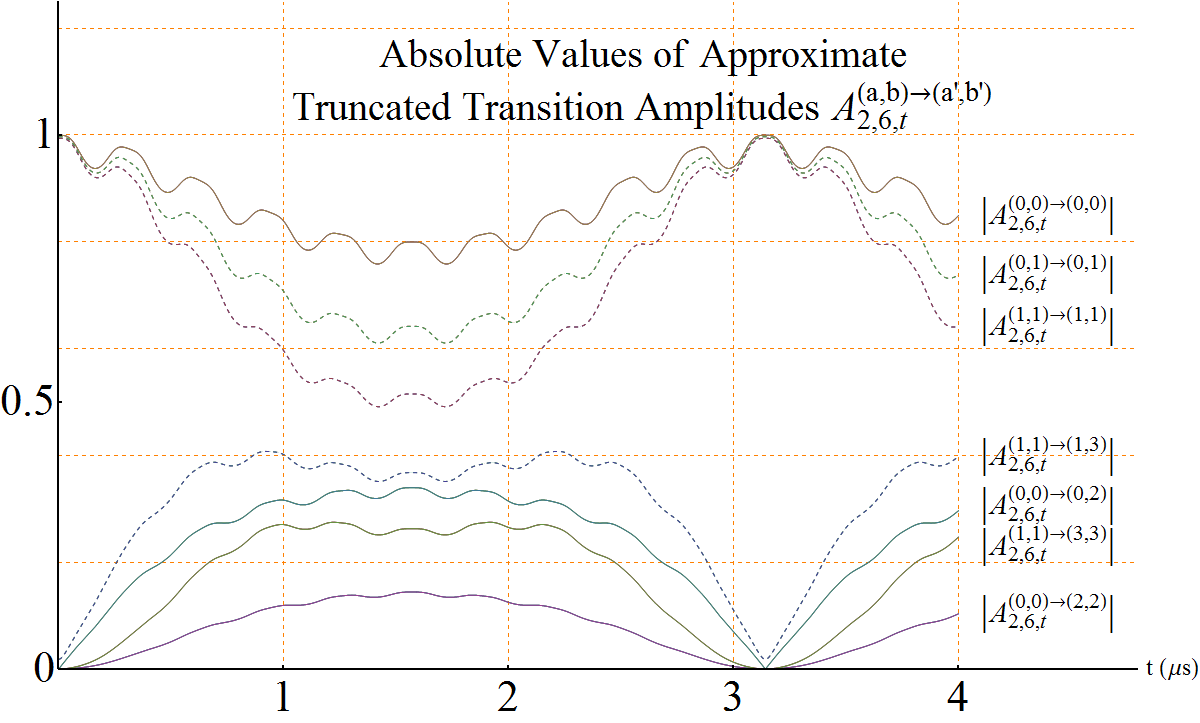}
      \caption{Absolute values of approximations to our truncated amplitudes are depicted for our model with parameters given by Table
  (\ref{table:qho-parameters}).  }\label{fig:truncated-dynamics}
\end{figure}

Let the spectral decomposition of $\chi$ have the form $\chi = \sum_{i=1}^4 \lambda_i |\lambda_i\> \<\lambda_i| $ where $|\lambda_i\>\<\lambda_i|$ are orthogonal projectors and the eigenvalues $\lambda_i$ are in non-increasing order \footnote{The matrix $\chi$ is symmetric and hence has real eigenvalues.}. For $\lambda_i \ge 0$, let $A_i$ be the image of the inverse map of the linear operator $|\cdot \>\>$ acting on 
$\sqrt{\lambda_{i}} |\lambda_{i}\>$. 
  Numerically diagonalizing the matrix $\chi$ shows that it has only one dominant eigenvalue. 
Hence we use the quantum operation $\cA(v) :=  A_1 v A_1^\dagger$ to approximate the truncated channel $\cN$ where 
 $A_1 = \sqrt{\lambda_1} (
   k_0 |0\>\<0|
 + k_2 |2\>\<0|
 + k_1 |1\>\<1|
 + k_3 |3\>\<1|)
 $ and $|\lambda_1\>  
= k_0 |0,0\> + k_2|2,0\> +
 k_1 |1,1\> + k_3|3,1\> 
 $. Observe that for $i,j \in \{0,1,2,3\}$ we have 
 \begin{align}
 \cA(|i\>\<j|) = \lambda_1 
 \bigl(k_i |i\> + k_{i+2} |i+2\> \bigr)
 \bigl(k_j^* \<j| + k_{j+2}^* \< j+2| \bigr).
 \end{align}

\subsubsection{Leakage error}
Leakage error, a dominant process of our model, occurs when low energy states transition into higher energy states within the quantum system. We derive lower bounds on the minimum amount of qubit leakage in our system, when the quantum channel is  $\Phi_t$ with Kraus set $\fK_t$ given by (\ref{eq:kraust}).
For qubit leakage to occur, it suffices to have the strict inequality  
\begin{align}
{\rm Leakage}(\Phi_t, \rho) := \sum_{i=2}^\infty \<i_\cX | \Phi_t(\rho) | i_\cX\> > 0, \label{ineq:leakage}
\end{align}
for some density operator $\rho$ supported on the span of $|0_\cX\>$ and $|1_\cX\>$. The above expression quantifies the amount of leakage from the qubit state space of our system.

We proceed to obtain a strictly positive lower bound on ${\rm Leakage}(\Phi_t, \rho)$. If $\Phi_t = \cN + \cM$ for some completely positive maps $\cN$ and $\cM$, then the complete positivity of $\cN$ and $\cM$ implies that 
$
{\rm Leakage}(\Phi_t, \rho) \ge {\rm Leakage}(\cN, \rho).
$
Hence it suffices to obtain a lower bound for ${\rm Leakage}(\cN, \rho)$. For our application, $\cN$ is our truncated channel. 

Assume that the initial state of the system is $\rho =  \half \bigl( | 0_\cX \> \<0_\cX| +  | 1_\cX \> \< 1_\cX| \bigr)$, the maximally mixed state in the qubit space. Then
the leakage of the truncated channel $\cN$ is at least
$
\half  \Bigl( T_{\Omega_{L,t}}^{(0,0) \to (2,2)}  +
T_{\Omega_{L,t}}^{(1,1) \to (3,3)} \Bigr)  .
$
Thus we have that
$
{\rm Leakage}(\cN, \rho)  \ge
\half  \Bigl( \bigl|  A_{{L,N,t}}^{(0,0) \to (2,2)} \bigr| +
\bigl| A_{{L,N,t}}^{(1,1) \to (3,3)}  \bigr| \Bigr)  - 2\eps
$
where $\eps$ is given the upper bound in Theorem 
\ref{thm:qho-error-bound}. In view of the data given in Figure \ref{fig:truncated-dynamics}, the amount of qubit leakage can actually be quite substantial. In particular, when $t= 5 \times 10^{-6}$, the amount of leakage of $\Phi_t$ is at least
$ 0.4$ which is substantially larger than zero.

The large amount of leakage from our qubit state highlights the importance of accounting for transitions of low energy states into excited states in oscillator systems, and also the problem of using the lowest two energy eigenstates as a basis to encode our qubit.

\subsubsection{Quantum error correction versus no quantum error correction}
We consider a universe with four harmonic oscillators $\cX_1,\cX_2, \cY_1$ and $\cY_2$. We identify the oscillators $\cX_1$ and $\cX_2$ as $x$-type oscillators and the oscillators $\cY_1$ and $\cY_2$ as $y$-type oscillators, with their parameters given by Table \ref{table:qho-parameters}. Assume that there are only $\cX_1$-$\cY_1$ couplings and $\cX_2$-$\cY_2$ couplings in our universe. Suppose that a maximally entangled two-qubit state is initialized in the $\cX_1$ and $\cX_2$ oscillators supported on their two lowest energy levels. We assume that the oscillators $\cY_1$ and $\cY_2$ are initialized in the ground state. We obtain a lower bound on the fidelity of the time-evolved states when instantaneous, identical and independent recovery operations are performed on oscillators $\cX_1$ and $\cX_2$. Denoting our truncated channel and recovery channel on a single system oscillator by $\cN$ and $\cR$ respectively, our lower bound on the output fidelity is 
\begin{align}
f_\cR = \<\Phi| (\cR \otimes \cR ) \circ (\cN \otimes \cN) (|\Phi\>\<\Phi|) |\Phi\>,
\end{align}
where 
$|\Phi\> =  (
|0_{\cX_1}, 0_{\cX_2}\> +
|1_{\cX_1}, 1_{\cX_2}\> ) /\sqrt{2}$. 
We now proceed to obtain a lower bound on the entanglement fidelity of our time-evolved maximally entangled state, with and without recovery. The
maximally entangled state on two qubits written as a density matrix is
\begin{align}
|\Phi\>\<\Phi| =\frac{1}{2} \sum_{i,j \in \{0,1\} } |i,i\>\<j,j| .
\end{align}
By linearity, the action of our truncated channels on the maximally entangled state gives
\begin{align}
(\cN \otimes \cN)|\Phi\>\<\Phi| 
= \frac{1}{2} \sum_{i,j \in \{0,1\} } \cN( |i\>\<j|) \otimes \cN( |i\>\<j|).
\end{align}
Let us denote the error of the truncated transition amplitude of $A_{iji'j'}$ to be $\eps_{iji'j'}$ so that 
\begin{align}
(\cN \otimes \cN)|\Phi\>\<\Phi| 
&= \frac{1}{2} \sum_{i,j \in \{0,1\} } \cN( |i\>\<j|) \otimes \cN( |i\>\<j|) \notag \\
&= \frac{1}{2} \sum_{\substack{
	i,j \in \{0,1\}\\
	i_1,j_1 \in \{0,1,2,3\}\\ 
	i_2,j_2 \in \{0,1,2,3\}\\ 
 }}
 	(A_{ij i_1 j_1} + \eps_{ij i_1 j_1} )
 	(A_{ij i_2 j_2} + \eps_{ij i_2 j_2} ) |i_1, i_2\>\<j_1, j_2| \label{eq:N_otimes_N_state}
\end{align}
If we perform no recovery operation, $\cR$ is just the identity map $\cI$, and we have
\begin{align}
f_\cI 
	&= \<\Phi| (\cN \otimes \cN) (|\Phi\>\<\Phi|) |\Phi\> \notag\\
	&= \<\Phi| 
	 \left( \sum_{(i,j)\in \{0,1\}}	\cN(|i_{\cX_1}\>\<j_{\cX_1}|) \otimes \cN(|i_{\cX_2}\>\<j_{\cX_2}|) \right)
	|\Phi\>/2 .
\end{align}
Dropping the labels on the Hilbert spaces of our bras and kets, we can use (\ref{eq:N_otimes_N_state}) to find that 
\begin{align}
f_\cI &=
   \frac{1}{2}( \<0, 0| + \<1, 1|)
   		(\cN \otimes \cN)|\Phi\>\<\Phi| 
    ( |0, 0\> + |1, 1\>) \notag\\
&= \frac{1}{4} \sum_{\substack{
	i,j \in \{0,1\}\\
	i_1,j_1 \in \{0,1,2,3\}\\ 
	i_2,j_2 \in \{0,1,2,3\}\\ 
 }}
 	(A_{ij i_1 j_1} + \eps_{ij i_1 j_1} )
 	(A_{ij i_2 j_2} + \eps_{ij i_2 j_2} ) |i_1, i_2\>\<j_1, j_2| \notag\\
 	&\ge
 	 \frac{1}{4} (A_{0000}^2 + A_{1111}^2 + A_{0101}^2 +A_{1010}^2) - 8(2\eps+\eps^2)
\end{align}
where $\eps \le 0.00084$. Note that $ A_{0101}^2 +A_{1010}^2$ can be negative.

The Barnum-Knill recovery operator $\cR^{\rm BK} $ \cite{BaK02} and the Tyson-Beny-Oreshkov quadratic recovery operator \cite{Tys10, BeO10} are near optimal recovery operators defined with respect to a quantum operation $\cA$ and a state $\rho$, and are equivalent when $\cA$ has only one Kraus operator. For our application, we study the Barnum-Knill recovery, with $\rho = (|0\>\<0| + |1\>\<1|)/2 $ and quantum operation $\cA$ approximating the truncated channel $\cN$. Our choice of $\rho$ shows that our quantum information is encoded in the trivial quantum code (no encoding).
For our application, the Barnum-Knill recovery operator which is also a quantum operation is defined as
\begin{align}
\cR^{\rm BK}(v) = 
A_1^\dagger 
	(A_1A_1^\dagger)^{-1/2^+}
	v
	(A_1A_1^\dagger)^{-1/2^+}
A_1
\end{align}
where $(A_1A_1^\dagger)^{^-1/2^+}$ is the square root of the psuedo-inverse of the operator $(A_1A_1^\dagger)$. When we use the Barnum-Knill recovery operation $\cR^{\rm BK}$, the fidelity of recovery is
\begin{align}
f_{\rm BK}
	&= \<\Phi|(\cR^{\rm BK}\otimes \cR^{\rm BK}) \circ (\cN \otimes \cN) (|\Phi\>\<\Phi|) |\Phi\> .
	\ge \lambda_1^2 - \frac{|\lambda_2| + |\lambda_3| + |\lambda_4|}{4} - \frac{1}{4} (2\eps + \eps^2).
\end{align}
We plot lower bounds of the fidelity with and without Barnum-Knill recovery in Figure \ref{fig:qho-fidelity}, and demonstrate that a fidelity of more than 60\% is still possible in spite of the leakage error and the use of truncated quantum channels in our analysis.
\begin{figure}[htb]
  \centering
    \includegraphics[width=1\textwidth]{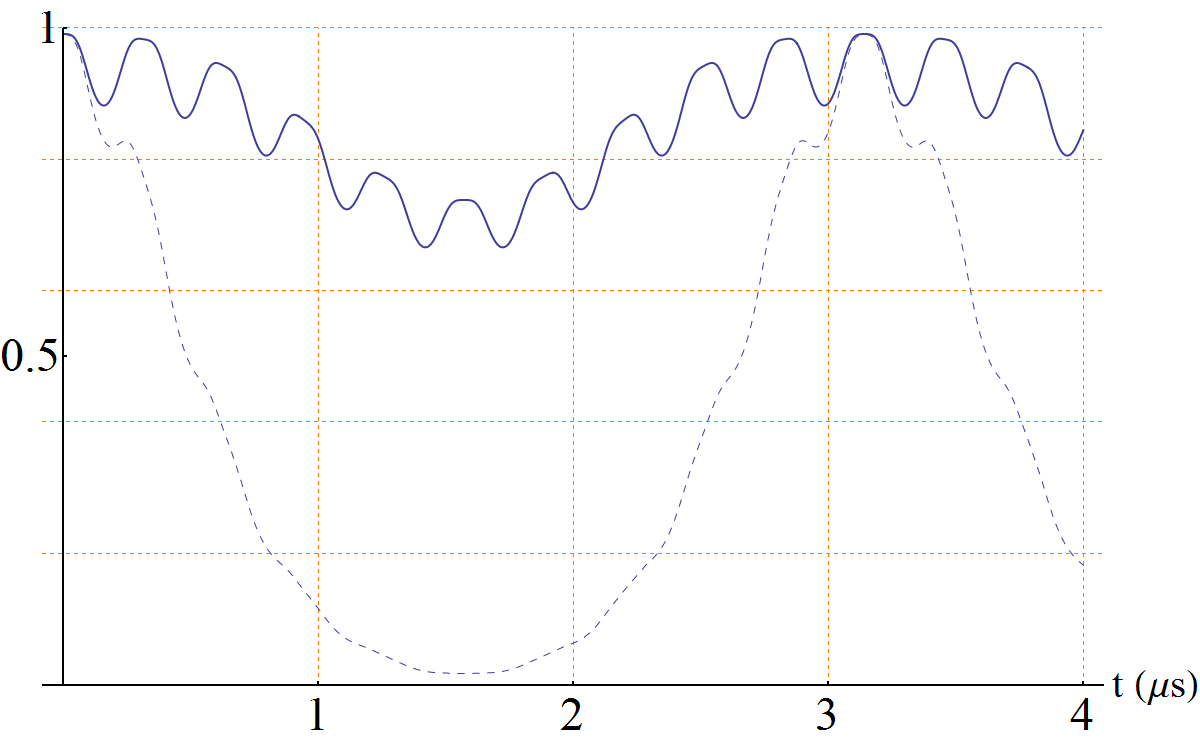}
      \caption{Lower bounds on the fidelity of an entangled state without recovery $f_\cI$ (dashed line) and with Barnum-Knill recovery $f_{\rm BK}$ (solid line) are plotted with respect to time.
      }\label{fig:qho-fidelity}
\end{figure}

\section{Discussions}
The system we consider is described by a quantum harmonic oscillator coupled through a spring-like interaction to another initially decoupled harmonic oscillator. We provide approximations to the truncated transition amplitudes of such a system. The converging error bound of such approximations is our main result. Properties of the integrals of products of Hermite functions lie at the heart of the proof. It is also worth noting that any ensemble of harmonic oscillators with spring-like coupling can be analyzed similarly.  

We also show numerically that in agreement with intuition from perturbation theory, when the oscillators are off-resonant and weakly coupled, amplitude damping is a negligible physical process. We also use our truncated channel representation to show that qubit leakage can be a dominant physical process, and how Barnum-Knill recovery can help protect a maximally entangled state stored in two oscillators each coupled independently to distinct zero-temperature harmonic baths, in the paradigm of off-resonant and weak coupling between the system and a zero temperature bath.

\section{Acknowledgements}
We like to thank Si-Hui Tan and the referees for helpful suggestions.

\bibliography{mybib}{}

\begin{thebibliography}{10}

\bibitem{FKM65}
G.~W. Ford, M.~Kac, and P.~Mazur, ``Statistical mechanics of assemblies of
  coupled oscillators,'' {\em Journal of Mathematical Physics}, vol.~6, p.~504,
  April 1965.

\bibitem{RLL67}
Z.~Rieder, J.~L. Lebowitz, and E.~Lieb, ``{Properties of a harmonic crystal in
  a stationary nonequilibrium state},'' {\em J. Math. Phys.}, vol.~8, p.~1073,
  1967.

\bibitem{EKN68}
L.~E. Estes, T.~H. Keil, and L.~M. Narducci, ``Quantum-mechanical description
  of two coupled harmonic oscillators,'' {\em Phys. Rev.}, vol.~175,
  pp.~286--299, Nov 1968.

\bibitem{BBW73}
T.~Banks, C.~M. Bender, and T.~T. Wu, ``Coupled anharmonic oscillators. i.
  equal-mass case,'' {\em Phys. Rev. D}, vol.~8, pp.~3346--3366, Nov 1973.

\bibitem{Dav73}
E.~B. Davis, ``The harmonic oscillator in a heat bath,'' {\em Commun. Math.
  Phys.}, pp.~171--186, 1973.

\bibitem{MiH86}
G.~J. Milburn and C.~A. Holmes, ``Dissipative quantum and classical {L}iouville
  mechanics of the anharmonic oscillator,'' {\em Phys. Rev. Lett.}, vol.~56,
  pp.~2237--2240, May 1986.

\bibitem{NRSS09}
B.~Nachtergaele, H.~Raz, B.~Schlein, and R.~Sims, ``Lieb-{R}obinson bounds for
  harmonic and anharmonic lattice systems,'' {\em Communications in
  Mathematical Physics}, vol.~286, pp.~1073--1098, 2009.
\newblock 10.1007/s00220-008-0630-2.

\bibitem{YUKG88}
K.-H. Yeon, C.-I. Um, W.-H. Kahng, and T.~F. George, ``Propagators for driven
  coupled harmonic oscillators,'' {\em Phys. Rev. A}, vol.~38, pp.~6224--6230,
  Dec 1988.

\bibitem{HPZ92}
B.~L. Hu, J.~P. Paz, and Y.~Zhang, ``Quantum brownian motion in a general
  environment: Exact master equation with nonlocal dissipation and colored
  noise,'' {\em Phys. Rev. D}, vol.~45, pp.~2843--2861, Apr 1992.

\bibitem{CYH08}
C.-H. Chou, T.~Yu, and B.~L. Hu, ``Exact master equation and quantum
  decoherence of two coupled harmonic oscillators in a general environment,''
  {\em Phys. Rev. E}, vol.~77, p.~011112, Jan 2008.

\bibitem{MaG12}
D.~X. Macedo and I.~Guedes, ``Time-dependent coupled harmonic oscillators,''
  {\em Journal of Mathematical Physics}, vol.~53, p.~052101, 2012.

\bibitem{Dav74}
E.~B. Davis, ``Markovian master equations,'' {\em Commun. Math. Phys.},
  pp.~91--110, 1974.

\bibitem{Dav69}
E.~B. Davis, ``Quantum stochastic processes,'' {\em Commun. Math. Phys.},
  pp.~277--304, 1969.

\bibitem{Dav70}
E.~B. Davis, ``Quantum stochastic processes {II},'' {\em Commun. Math. Phys.},
  pp.~83--105, 1970.

\bibitem{Dav71}
E.~B. Davis, ``Quantum stochastic processes {III},'' {\em Commun. Math. Phys.},
  pp.~51--70, 1971.

\bibitem{RMa81}
R.~Benguria and M.~Kac, ``Quantum {L}angevin equation,'' {\em Phys. Rev.
  Lett.}, vol.~46, pp.~1--4, Jan 1981.

\bibitem{FLO88}
G.~W. Ford, J.~T. Lewis, and R.~F. O'Connell, ``Quantum {L}angevin equation,''
  {\em Phys. Rev. A}, vol.~37, pp.~4419--4428, Jun 1988.

\bibitem{Kos72}
A.~Kossakowski, ``On quantum statistical mechanics of non-{H}amiltonian
  systems,'' {\em Reports on Mathematical Physics}, vol.~3, no.~4, pp.~247 --
  274, 1972.

\bibitem{Lin76}
G.~Lindblad, ``On the generators of quantum dynamical semigroups,'' {\em
  Communications in Mathematical Physics}, vol.~48, pp.~119--130, 1976.
\newblock 10.1007/BF01608499.

\bibitem{parr1994density}
R.~Parr, {\em Density-Functional Theory of Atoms and Molecules}.
\newblock International Series of Monographs on Chemistry, Oxford University
  Press, USA, 1994.

\bibitem{Kos83}
D.~Kosloff and R.~Kosloff, ``A fourier method solution for the time dependent
  schrödinger equation as a tool in molecular dynamics,'' {\em Journal of
  Computational Physics}, vol.~52, no.~1, pp.~35 -- 53, 1983.

\bibitem{LCDFGZ87}
A.~J. Leggett, S.~Chakravarty, A.~T. Dorsey, M.~P.~A. Fisher, and W.~Z.
  Anupam~Garg, ``Dynamics of the dissipative two-state system,'' {\em Reviews
  of Modern Physics}, vol.~59, no.~1, 1987.
\newblock quant-ph/0610063.

\bibitem{nielsen-chuang}
M.~A. Nielsen and I.~L. Chuang, {\em Quantum Computation and Quantum
  Information}.
\newblock Cambridge University Press, second~ed., 2000.

\bibitem{Choi75}
M.-D. Choi, ``Completely positive linear maps on complex matrices,'' {\em
  Linear Algebra and its Applications}, vol.~10, no.~3, pp.~285 -- 290, 1975.

\bibitem{LNCY97}
D.~W. Leung, M.~A. Nielsen, I.~L. Chuang, and Y.~Yamamoto, ``Approximate
  quantum error correction can lead to better codes,'' {\em Phys. Rev. A},
  vol.~56, p.~2567, 1997.

\bibitem{BaK02}
H.~Barnum and E.~Knill, ``Reversing quantum dynamics with near-optimal quantum
  and classical fidelity,'' {\em Journal of Mathematical Physics}, vol.~43,
  p.~2097, Jan 2002.

\bibitem{Fletcher08}
A.~Fletcher, P.~Shor, and M.~Win, ``Channel-adapted quantum error correction
  for the amplitude damping channel,'' {\em Information Theory, IEEE
  Transactions on}, vol.~54, pp.~5705 --5718, dec. 2008.

\bibitem{Kosut08}
R.~L. Kosut, A.~Shabani, and D.~A. Lidar, ``Robust quantum error correction via
  convex optimization,'' {\em Phys. Rev. Lett.}, vol.~100, p.~020502, Jan 2008.

\bibitem{BaG09}
G.~Ball\'o and P.~Gurin, ``Robustness of channel-adapted quantum error
  correction,'' {\em Phys. Rev. A}, vol.~80, p.~012326, Jul 2009.

\bibitem{Tys10}
J.~Tyson, ``Two-sided bounds on minimum-error quantum measurement, on the
  reversibility of quantum dynamics, and on maximum overlap using directional
  iterates,'' {\em Journal of Mathematical Physics}, vol.~51, p.~092204, Jun
  2010.

\bibitem{BeO10}
C.~B\'eny and O.~Oreshkov, ``General conditions for approximate quantum error
  correction and near-optimal recovery channels,'' {\em Phys. Rev. Lett.},
  vol.~104, p.~120501, Mar 2010.

\bibitem{BeO11}
C.~B\'eny and O.~Oreshkov, ``Approximate simulation of quantum channels,'' {\em
  Phys. Rev. A}, vol.~84, p.~022333, Aug 2011.

\bibitem{Ouyang-PKL}
Y. Ouyang, ``The Perturbed Error-Correction Criterion and Rescaled Truncated
  Recovery'', in preparation.

\bibitem{CLY97}
I.~L. Chuang, D.~W. Leung, and Y.~Yamamoto, ``Bosonic quantum codes for
  amplitude damping,'' {\em Phys. Rev. A}, vol.~56, p.~1114, 1997.

\bibitem{Liu04}
Y.-x. Liu, S.~K. \"Ozdemir, A.~Miranowicz, and N.~Imoto, ``Kraus representation
  of a damped harmonic oscillator and its application,'' {\em Phys. Rev. A},
  vol.~70, p.~042308, Oct 2004.

\bibitem{Hol11}
A.~S. Holevo, ``The {C}hoi–-{J}amiolkowski forms of quantum {G}aussian
  channels,'' {\em Journal of Mathematical Physics}, vol.~52, no.~4, p.~042202,
  2011.

\bibitem{reed+simon-I}
M.~Reed and B.~Simon, {\em Methods of Mathematical Physics I: Functional
  Analysis}.
\newblock New York and London: Academic Press, first~ed., 1972.

\bibitem{Bogo+L+T}
I.~T.~T. Nikolai Nikolaevich~Bogolubov, Anatolii Alekseevich~Logunov, {\em
  Introduction to Axiomatic Quantum Field Theory}.
\newblock Massachusetts 01867, U.S.A: W. A. Benjamin, Inc., first~ed., 1975.

\bibitem{HeK69}
K.~E. Hellwig and K.~Kraus, ``Pure operations and measurements,'' {\em
  Communications in Mathematical Physics}, vol.~11, pp.~214--220, 1969.
\newblock 10.1007/BF01645807.

\bibitem{HeK70}
K.~E. Hellwig and K.~Kraus, ``Operations and measurements. {II},'' {\em
  Communications in Mathematical Physics}, vol.~16, pp.~142--147, 1970.
\newblock 10.1007/BF01646620.

\bibitem{kraus}
K.~Kraus, {\em Lecture Notes in Physics 190 : States, Effects, and Operations
  Fundamental Notions of Quantum Theory}.
\newblock Springer Berlin / Heidelberg, first~ed., 1983.

\bibitem{abramowitz+stegun}
M.~Abramowitz and I.~A. Stegun, {\em Handbook of Mathematical Functions with
  Formulas, Graphs, and Mathematical Tables}.
\newblock New York: Dover, ninth dover printing, tenth gpo printing~ed., 1964.

\bibitem{Wat33}
G.~N. Watson, ``Notes on generating functions of polynomials: (2) {H}ermite
  polynomials,'' {\em Journal of the {L}ondon Mathematical Society}, no.~3,
  pp.~194 -- 199, 1933.

\bibitem{wong}
M.~W. Wong, {\em Weyl Transforms}.
\newblock New York: Springer, 1998.

\bibitem{shankar}
R.~Shankar, {\em Principles of Quantum Mechanics}.
\newblock 233 Spring Street, New York, N.Y. 100130: Plenum Press, second~ed.,
  1994.

\bibitem{JMM11}
A.~Jellal, F.~Madouri, and A.~Merdaci, ``Entanglement in coupled harmonic
  oscillators studied using a unitary transformation,'' {\em Journal of
  Statistical Mechanics: Theory and Experiment}, vol.~2011, no.~09, p.~P09015,
  2011.

\bibitem{ChC07}
N.~N. Chung and L.~Y. Chew, ``Energy eigenvalues and squeezing properties of
  general systems of coupled quantum anharmonic oscillators,'' {\em Phys. Rev.
  A}, vol.~76, p.~032113, Sep 2007.

\bibitem{ChC09}
N.~N. {Chung} and L.~Y. {Chew}, ``{Two-step approach to the dynamics of coupled
  anharmonic oscillators},'' {\em Phys. Rev. A}, vol.~80, p.~012103, July 2009.

\bibitem{Alz09}
H.~Alzer, ``Sharp upper and lower bounds for the gamma function,'' {\em
  Proceedings of the Royal Society of Edinburgh: Section A Mathematics},
  vol.~139, pp.~709--718, August 2009.

\bibitem{KHB09}
M.~L.~H. R.~Kerman and M.~Brannan, ``Error estimates for dominici's hermite
  function asymptotic formula and some applications,'' {\em The ANZIAM
  Journal}, vol.~50, pp.~550--561, April 2009.

\bibitem{erdelyi}
A.~Erd\'{e}lyi, {\em Higher Transcendental Functions 2}.
\newblock Robert E. Krieger Publishing Company, second reprint~ed., 1985.

\end{thebibliography}
\addcontentsline{toc}{paper}{Bibliography}
\bibliographystyle{ieeetr}
\end{document}